\documentclass[journal]{IEEEtran}

\IEEEoverridecommandlockouts

\usepackage{textcomp}
\usepackage{xcolor}
\usepackage{graphicx} 
\usepackage{epstopdf}
\usepackage{cite}
\usepackage{graphicx}
\usepackage{amsmath}
\usepackage{amssymb}
\usepackage{amsfonts}
\usepackage{array}

\usepackage{color}
\usepackage{algorithm, algorithmic}
\usepackage{booktabs} 
\usepackage{multirow}
\usepackage{url}
\usepackage{etoolbox}
\usepackage{amsthm} 

\newtheorem{proposition}{Proposition}
\newtheorem{theorem}{Theorem}

\usepackage{subcaption}
\usepackage{graphicx,bm,dsfont}

\newcommand{\todom}[1]{} 

\def\BibTeX{{\rm B\kern-.05em{\sc i\kern-.025em b}\kern-.08em
    T\kern-.1667em\lower.7ex\hbox{E}\kern-.125emX}}

\begin{document}

\title{A Unified Distributed Algorithm for Hybrid Near-Far Field Activity Detection in Cell-Free Massive MIMO\\
}

\author{\IEEEauthorblockN{Jingreng Lei,  Yang Li, Ziyue Wang, Qingfeng Lin, Ya-Feng Liu, and Yik-Chung Wu}

\thanks{Received 18 September 2025; revised 26 March 2026; accepted 23 May 2026.  The work of Jingreng Lei and Yang Li was supported in part by the National Natural Science Foundation of China (NSFC) under Grant 62571086, and in part by Guangdong Basic and Applied Basic Research Foundation under Grant 2025A1515011658. The work of Ya-Feng Liu was supported in part by the NSFC under Grant 12371314. An earlier version of this paper was presented at the IEEE Globecom 2025. [DOI: 10.1109/GLOBECOM59602.2025.11431942]. The associate editor coordinating the review of this article and approving it for publication was Prof. Theodoros Tsiftsis. (Corresponding author: Yang Li.)}
\thanks{Jingreng Lei is with the Department of Electrical and Computer Engineering, The University of Hong Kong, Hong Kong, and also with the School of Computing and Information Technology, Great Bay University, Dongguan~523000, China (e-mail: leijr@eee.hku.hk).}
\thanks{Yang Li is with the School of Computing and Information Technology, Great Bay University, Dongguan 523000, China, and also with Dongguan Key Laboratory for Intelligence and Information Technology, Dongguan 523000, China (e-mail: liyang@gbu.edu.cn).}
\thanks{Ziyue Wang is with the State Key Laboratory of Scientific and Engineering Computing, Institute of Computational Mathematics and Scientific/Engineering Computing, Academy of Mathematics and Systems Science, Chinese Academy of Sciences, Beijing 100190, China (e-mails: ziyuewang@lsec.cc.ac.cn).}
\thanks{Qingfeng Lin and Yik-Chung Wu are with the Department of Electrical and Computer Engineering, The University of Hong Kong, Hong Kong (e-mail: \{qflin, ycwu\}@eee.hku.hk). }
\thanks{Ya-Feng Liu is with the Ministry of Education Key Laboratory of Mathematics and Information Networks, School of Mathematical Sciences, Beijing University of Posts and Telecommunications, Beijing 102206, China (e-mail: yafengliu@bupt.edu.cn).}
}

\maketitle

\begin{abstract} 
A great amount of endeavor has recently 
been devoted to activity detection for 
massive machine-type communications in 
cell-free multiple-input multiple-output (MIMO) systems. However, as the number of antennas at the  
access points (APs) increases, the Rayleigh 
distance that separates the near-field and 
far-field regions also expands, rendering 
the conventional assumption of sole far-field 
propagation impractical. To address 
this challenge, this paper establishes a  
covariance-based formulation that can effectively capture the statistical property of hybrid near-far field channels. Based on this formulation, we theoretically reveal that increasing the proportion of near-field channels enhances the detection performance. Furthermore,  we propose a distributed algorithm, where each AP performs local activity detection and only exchanges the detection results to the central processing unit, thus significantly reduces the computational complexity and the communication overhead. Not only with convergence guarantee, the proposed algorithm is unified in the sense that it can handle single-cell or cell-free systems with either near-field or far-field devices as special cases.
Simulation results validate 
the theoretical analyses and demonstrate 
the superior performance of the proposed  
approach compared with existing methods.
\end{abstract}

\begin{IEEEkeywords}
Cell-free massive MIMO, distributed activity detection, grant-free random access, hybrid near-far field communications, machine-type communications.
\end{IEEEkeywords}

\section{Introduction}
With the rapid development of the Internet of Things (IoT), massive machine-type communications (mMTC) are expected to play a crucial role in the sixth-generation (6G) vision of ubiquitous connectivity\cite{Yafeng}. To meet stringent low-latency requirements, grant-free random access has emerged as a promising solution\cite{shahab2020grant,yangge}, where devices transmit data without obtaining permission from the access points (APs). However, a key challenge in grant-free random access lies in the device activity detection, as large number of potential IoT devices prevents orthogonality among the signature sequences of different devices. 

Current studies on activity detection can be broadly divided into two categories. In the first line of research, by exploiting the sporadic nature of mMTC, compressed sensing (CS)-based methods can be employed to solve the joint activity detection and channel estimation problem~\cite{yangge3,gao2023compressive,zhanghao,Chen2018,Ke20,Chen2019,Senel2018,Mei21,Ai22,zhanghao2}. Another line of research, known as the covariance-based approach, identifies active devices by leveraging the channel covariance matrix without explicitly estimating the channels\cite{Haghighatshoar2018,chenIcc,Lin2022,li2022asynchronous,lin2024intelligent,Ren2025}. Both theoretically and empirically, it has been demonstrated that the covariance-based approach generally outperforms the \mbox{CS-based} approach\cite{chen2021phase}.

Activity detection problem has been studied in both single cell scenario as well as the more general cell-free massive multiple-input multiple-output (MIMO), where all APs are connected to a central processing unit (CPU) via fronthaul links for joint signal processing\cite{Shao2020-2,ganesan2021clustering,lin2024communication}. The cell-free  architecture eliminates traditional cell boundaries, and multiple received signals of a user can be harnessed to improve detection performance compared to the single cell setting. Nevertheless, as the array aperture enlarges, the boundary between the near-field and the far-field regions, characterized by the Rayleigh distance, also expands~\cite{lu2024tutorial,cuilaoshi,wangzhe,liu2025nearfield,lu2025block}. \textcolor{black}{This leads to a  scenario, previously unexamined, in which  a device may appear as near-field to certain APs and as far-field to other APs, as shown in Fig.~\ref{fig:Hybrid}.} Due to the potential mixed identity of a user, the conventional assumption of the 
far-field propagation alone in existing activity detection\cite{yangge3,gao2023compressive,zhanghao,Chen2018,Ke20,Chen2019,Senel2018,Mei21,Ai22,zhanghao2,Haghighatshoar2018,chenIcc,Lin2022,li2022asynchronous,lin2024intelligent,Ren2025} becomes inapplicable. \textcolor{black}{In particular, unlike far-field channels that are formulated as planar wavefronts and exhibit independent and identically distributed (i.i.d.) entries across antennas, near-field channels are with spherical wavefronts, where the channel depends on both the angle and the distance of the device~\cite{Cui1,Cui2,Cui3}. Therefore, the activity detection problem in the presence of near-field channels becomes substantially different from the conventional far-field setting and calls for a new treatment.}

Given the superior performance of the covariance-based approach for activity detection in conventional far-field cell-free massive MIMO\cite{ganesan2021clustering,li2022asynchronous,lin2024communication}, this paper aims to investigate the covariance-based approach when hybrid near-far field channels are present. The essential step for developing the covariance-based approach is to derive the probability density function (PDF) of the received signals. However, unlike the conventional far-field scenario where the received signal at each antenna is i.i.d. Gaussian distributed\cite{fengler2021non}, the hybrid near-far field channels are more complicated, making the channel covariance matrix very different from that in existing works~\cite{Haghighatshoar2018,chenIcc,Lin2022,li2022asynchronous,lin2024intelligent,Ren2025}. Consequently, the rank-one update commonly 
employed in the covariance-based method\cite{ganesan2021clustering,li2022asynchronous,wang2024scalinglaw} cannot be applied. Furthermore, as a particular device may be a near-field user to some APs but a
far-field user to other APs, it is not straightforward   to design a distributed algorithm that can leverage the parallel computation across multiple APs.

\begin{figure}[t!]
    \centering
    \includegraphics[width=0.8\linewidth]{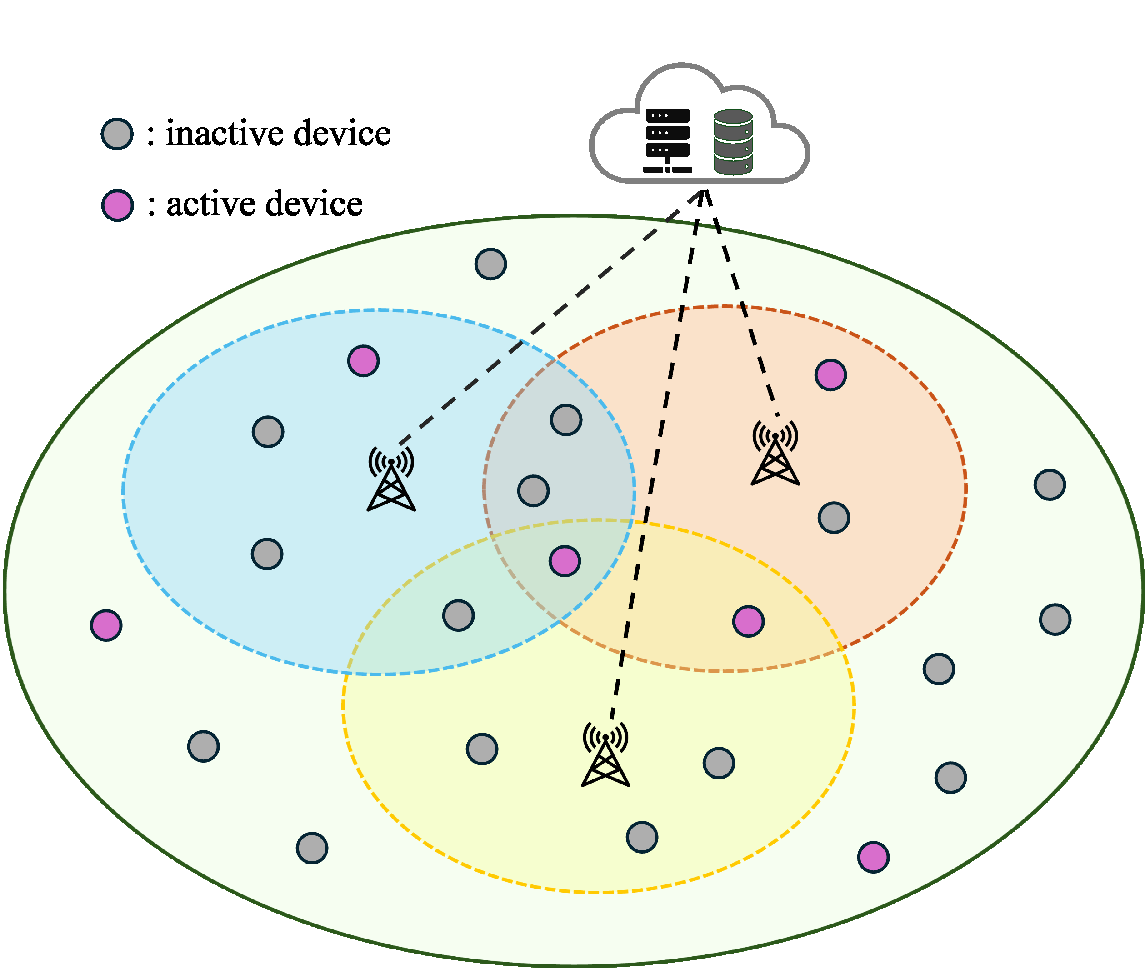}
    \caption{\textcolor{black}{Hybrid near-far field activity detection in cell-free massive MIMO.}}
    \label{fig:Hybrid}
\end{figure}  

Beyond the problem formulation and algorithm design, another fundamental theoretical question emerges: \textcolor{black}{how do hybrid near-far field channels affect the detection performance?} While the  near-field and far-field regions exhibit different electromagnetic propagation characteristics, existing works on activity detection \cite{Haghighatshoar2018,wang2024scalinglaw,chen2021phase, lin2024intelligent} merely focus on the performance analysis for the conventional far-field communications. This necessitates the study of the detection performance when hybrid near-far field channels are involved.

Despite the above challenges, this paper presents a systematic study on the hybrid near-far field activity detection in cell-free massive MIMO. In particular, we first derive the PDF of the received signals so that the statistical property of the hybrid near-far field channels is accurately captured. The newly established received signal model unifies various special cases such as single-cell or cell-free case with either near-field or far-field devices. Applying the established signal model to the covariance-based activity detection formulation, we theoretically analyze how the hybrid near-far field channels affect the detection performance from the perspective of identifiability. This insight provides a fundamental understanding of hybrid near-far field communications in the context of activity detection.

In addition to theoretical insight, we also propose a unified distributed algorithm to solve the formulated activity detection problem. Specifically, we introduce dual variables and penalty term such that the whole-network joint detection problem can be decomposed into multiple parallel detection subproblems, each executed at an AP. To efficiently handle the local detection problem at each AP, a novel coordinate descent (CD) algorithm based on the Sherman-Morrison-Woodbury update with Taylor expansion is developed. Furthermore, we provide theoretical convergence analysis to show that the proposed distributed algorithm is guaranteed to converge to a stationary point of the original centralized covariance-based activity detection problem.  Notably, by simply changing the system parameters
(e.g., number of APs, number of antennas, wavelength of
the carrier), the proposed distributed algorithm directly covers various scenarios, including the activity detection for the single-cell near-field or far-field, and the cell-free near-field or far-field communications. This versatility makes the proposed algorithm widely applicable in various communication scenarios.  

With the information exchange between APs and the CPU only involving local detection results instead of the original received signals, the  computational complexity and communication overhead of the proposed distributed algorithm is significantly reduced compared to the centralized approach. This is shown by theoretical analysis and confirmed by simulations, where the proposed algorithm achieves detection performance close to that of
the centralized approach while requiring much fewer bits for
fronthaul transmission and significantly shorter computation
time. Furthermore, simulation results show superior performance of the
proposed algorithm compared with existing benchmarks, and corroborate the theoretical analysis that increasing the proportion of near-field
channels leads to better detection performance. 

The rest of the paper is organized as follows. Section~\ref{sec:system} introduces the hybrid near-far field system model and problem formulation. Section~\ref{sec:detection_performance} theoretically analyzes how hybrid near-far field channels impact the detection performance. Section~\ref{sec:distributed} presents the unified distributed algorithm.  Section~\ref{convergence analysis} establishes the convergence guarantee of the proposed distributed algorithm. Simulation results are provided in Section~\ref{sec:simulation}. Finally, the paper is concluded in Section~\ref{sec:conclusion}.

\section{System Model and Problem Formulation}
\label{sec:system}
\subsection{System Model}

\textcolor{black}{We consider an uplink cell-free massive MIMO system with \( M \) APs and \( N \) IoT devices, as depicted in Fig.~\ref{fig:Hybrid}}. Each AP is equipped with a $K$-antenna uniform linear array with a half-wavelength space, and each IoT device utilizes a single antenna. All \( M \) APs communicate with a CPU through fronthaul links. The channels 
between the APs and the devices are assumed to undergo quasi-static block fading, which  remains unchanged within each coherence block, but may vary across different blocks. \textcolor{black}{The Rayleigh distance is defined as $2D^2/\lambda_\text{c}$, where $D$ is the aperture of the antenna array and $\lambda_\text{c}$ is the wavelength of the carrier. Devices located within this distance from an AP are considered near-field users, while the remaining devices are considered far-field users.} Notice that a device can appear as a near-field user for some APs and far-field user for other APs at the same time. 

For each AP $m$, let $\mathcal{U}_m$ denote the subset containing  $N_{m,\text{near}}$  devices in the near-field region, and $\mathcal{U}^{\text{c}}_m$ denote the complement set containing the remaining $N - N_{m,\text{near}}$ devices in the far-field region. If  device $n$ is located in the far-field region of  AP $m$,  the uplink channel  can be modeled as\cite{ganesan2021clustering}:
\begin{equation}
\label{far-field chanel}
   \mathbf{h}_{m,n}\sim\mathcal{CN}(\mathbf{0},g_{m,n}\mathbf{I}_{K}),~\forall~n \in \mathcal{U}^{\text{c}}_m,
\end{equation}
 where $g_{m,n}$ is 
the large-scale fading coefficient. On the other hand, if   device $n$ is located in the near-field region of AP $m$, its uplink channel consists of a deterministic line-of-sight (LoS) channel and a statistical multi-path non-line-of-sight (NLoS) channel induced by $L_{m}$ scatters\cite{wang2024beamfocusing}. Then, the near-field channel can be modeled as~\cite[Eq.~(132)]{liuNearFieldCommunicationsTutorial2023}:
 \begin{align}
	\label{near-field channel model}
    \mathbf{h}_{m,n} \!\!=\! \underbrace{\beta_{m,n} \mathbf{b}(\mathbf{r}_{m,n})}_{\text{LoS}} \!+ \sum_{\ell=1}^{L_{m}}\underbrace{ \varphi_{m,\ell}\tilde{\beta}_{m,n,\ell} \mathbf{b}(\tilde{\mathbf{r}}_{m,\ell})}_{\text{NLoS}}, \forall~n\in\mathcal{U}_m,
 \end{align}
 \noindent where $\beta_{m,n}$ and $\tilde{\beta}_{m,n,\ell}$ denote the LoS and the NLoS channel gains, respectively,  $\varphi_{m,l}\sim\mathcal{CN}(0,\sigma_{m,l}^2)$ is the reflection coefficient of the $l$-th scatterer of AP $m$ with variance $\sigma_{m,l}^2$, $\mathbf{r}_{m,n} $ denotes the distances between device $n$ and all the $K$ antennas of AP $m$, $\tilde{\mathbf{r}}_{m,l}$ denotes the distances between scatter $l$ and the $K$ antennas of AP $m$, and $\mathbf{b}(\cdot)$ denotes the near-field array response, which is given by
 \begin{align}
    \mathbf{b}(\mathbf{r}) =
    \big[ & e^{-j\frac{2\pi}{\lambda_\text{c}}r_1}, 
          e^{-j\frac{2\pi}{\lambda_\text{c}}r_2}, \dots, e^{-j\frac{2\pi}{\lambda_\text{c}}r_K} \big]^{\text{T}},
\end{align}

\noindent where $\mathbf{r}\triangleq[r_1, r_2,..., r_K]^\text{T}$ denotes \textcolor{black}{a $K$-length distance vector determined by the positions of the source and the antenna array elements, applicable to both $\mathbf{r}_{m,n}$ and $\tilde{\mathbf{r}}_{m,\ell}$ in~(\ref{near-field channel model}).} 

\textcolor{black}{In this paper, we consider a typically stationary scenario of mMTC \cite{Edith1},~\cite{Edith2}, where the IoT devices are stationary. This is reasonable in many practical scenarios, with IoT devices commonly deployed in fixed locations to monitor or
control specific physical phenomena, such as temperature, humidity, and structural health~\cite{Edith1}. Once deployed, they remain in fixed positions to ensure consistent and accurate data collection. Furthermore, a stationary deployment simplifies the setup and maintenance of IoT networks, which requires less frequent physical intervention for maintenance or replacement~\cite{Edith2}. Therefore, the locations of the devices and the scatterers can be acquired through environment sensing or calibration procedures \cite{Wymeersch2020} during the initial network registration phase. Based on the location information, the channel gains can be computed from the path-loss model \cite{liuNearFieldCommunicationsTutorial2023}, \cite{Dong2022}, \cite{Wang2026}.}

To detect the device activity, each device $n$ is pre-assigned a distinct signature sequence with length $L$, i.e., $\mathbf{s}_n \in \mathbb{C}^L$.  Given the sporadic nature of mMTC, only a small fraction of the  $N$  devices remain active within each coherence block. Let $a_n\in \{0,1\}$ denote the binary activity variable for device $n$ ($a_n= 1$ means user $n$ is active).  Then, the received signal at AP \( m \) can be written as

\begin{align}
\label{receive signal}
\mathbf{Y}_m & = \sum_{n \in \mathcal{U}_m} a_n \mathbf{s}_n \mathbf{h}_{m,n}^\text{T} + \sum_{n \in \mathcal{U}_m^{\text{c}}} a_n \mathbf{s}_n \mathbf{h}_{m,n}^\text{T} + \mathbf{W}_m,
\end{align}

\noindent where  \( \mathbf{W}_m\in \mathbb{C}^{L\times K}\) is the additive i.i.d. Gaussian noise at AP \( m \) with each entry following $\mathcal{CN}(0,\varsigma^2_m)$ and $\varsigma^2_m$ denoting the noise variance.

\subsection{Problem Formulation}
\textcolor{black}{Based on the location information, we can calculate the covariance matrix for the near-field channel in~(\ref{near-field channel model}) as}
\textcolor{black}{
\begin{align}
	\label{near-field channel covariance}
	\mathbf{R}_{m,n} &= \mathbb{E}\!\left[\left(\mathbf{h}_{m,n} - \beta_{m,n}\mathbf{b}(\mathbf{r}_{m,n})\right)\!\left(\mathbf{h}_{m,n} - \beta_{m,n}\mathbf{b}(\mathbf{r}_{m,n})\right)^{\!\text{H}}\right]\nonumber \\
	&= \sum_{\ell=1}^{L_m}\sum_{\ell'=1}^{L_m} \tilde{\beta}_{m,n,\ell}\tilde{\beta}_{m,n,\ell'}^{*}\, \mathbb{E}[\varphi_{m,\ell}\varphi_{m,\ell'}^{*}]\, \mathbf{b}(\tilde{\mathbf{r}}_{m,\ell})\mathbf{b}^{\text{H}}(\tilde{\mathbf{r}}_{m,\ell'})\nonumber \\
	&= \sum_{\ell=1}^{L_m} \sigma_{m,\ell}^2\, |\tilde{\beta}_{m,n,\ell}|^2\, \mathbf{b}(\tilde{\mathbf{r}}_{m,\ell})\mathbf{b}^{\text{H}}(\tilde{\mathbf{r}}_{m,\ell}),~\forall~n\in\mathcal{U}_m,
\end{align}
where the second equality follows by substituting the NLoS component of~(\ref{near-field channel model}) and noting that $\tilde{\beta}_{m,n,\ell}$ and $\mathbf{b}(\tilde{\mathbf{r}}_{m,\ell})$ are deterministic. The third equality uses the fact that $\mathbb{E}[\varphi_{m,\ell}\varphi_{m,\ell'}^{*}] = \sigma_{m,\ell}^2$ if $\ell = \ell'$ and $0$ otherwise.} \textcolor{black}{According to~(\ref{near-field channel model}), the only randomness in the near-field channel comes from the reflection coefficient $\varphi_{m,\ell}$, which is a complex Gaussian random variable\cite{liuNearFieldCommunicationsTutorial2023}, \cite{Zou2024},\cite{Bjornson2017}. Therefore, as a linear combination of Gaussian random variables, the near-field channel is a complex Gaussian vector:}
\begin{align}
    \label{near-field channel}
    \mathbf{h}_{m,n} \sim \mathcal{CN}(\beta_{m,n} \mathbf{b}(\mathbf{r}_{m,n}), \mathbf{R}_{m,n}),~\forall~n\in\mathcal{U}_m.
\end{align}
Given that the far-field channels and noise are also complex Gaussian random variables, while $\{{a_n}\}_{n=1}^{N}$ can be viewed as deterministic unknown, the received signal $\mathbf{Y}_m$ in (\ref{receive signal}) follows a complex Gaussian distribution as well. 

However, since $\mathbf{R}_{m,n}$ for near-field device in (\ref{near-field channel covariance}) is not diagonal, this makes  the columns of $\mathbf{Y}_m$   not independent. The joint distribution of  $\mathbf{Y}_m$ is therefore not the product of the distributions from each individual column. Thus, we consider its vectorized form $\mathbf{y}_m = \text{vec}(\mathbf{Y}_m)$, which can be written as
\begin{align}
	\label{vectorized form}
\mathbf{y}_m &= \sum_{n \in \mathcal{U}_m} a_n \mathbf{h}_{m,n} \otimes \mathbf{s}_n +\sum_{n \in \mathcal{U}_m^\text{c}} a_n \mathbf{h}_{m,n} \otimes \mathbf{s}_n + \mathbf{w}_m,
\end{align}

\noindent \textcolor{black}{where $\otimes$ denotes the Kronecker product, which vectorizes each $L\times K$ outer product $\mathbf{s}_n \mathbf{h}_{m,n}^\text{T}$ in (\ref{receive signal}) into an $LK$-dimensional vector, and \( \mathbf{w}_m \) is the vectorized noise.} From (\ref{vectorized form}), we have $\mathbf{y}_m\sim\mathcal{CN}(\bar{\mathbf{y}}_{m},\mathbf{C}_{m})$, where the  mean and the covariance matrix are calculated based on (\ref{far-field chanel}) and (\ref{near-field channel}), which are given by
\begin{align}
\label{expectation and covariance}
    \bar{\mathbf{y}}_{m} &=  \sum_{n \in \mathcal{U}_m} a_n \left(\beta_{m,n} \mathbf{b}(\mathbf{r}_{m,n})\right) \otimes \mathbf{s}_n, \notag \\
    \mathbf{C}_{m} &= \sum_{n \in \mathcal{U}_m} a_n \mathbf{R}_{m,n} \otimes (\mathbf{s}_n \mathbf{s}_n^\text{H}) + \sum_{n \in \mathcal{U}_m^\text{c}} a_n g_{m,n} \mathbf{I}_K \otimes (\mathbf{s}_n \mathbf{s}_n^\text{H}) \notag \\
    &\quad\quad + \varsigma^2_m\mathbf{I}_{LK}.
\end{align}

\noindent  Then, the joint PDF of the received signals \( \{\mathbf{y}_m\}_{m=1}^M \)  is given by
\begin{align}
\label{joint PDF}
&p( \{\mathbf{y}_m\}_{m=1}^M ; \mathbf{a})\notag\\
 = &\prod_{m=1}^M \frac{1}{\left|\pi \mathbf{C}_{m}\right |} \exp\left(-(\mathbf{y}_m - \bar{\mathbf{y}}_{m})^\text{H} \mathbf{C}_{m}^{-1} (\mathbf{y}_m - \bar{\mathbf{y}}_{m})\right),
\end{align}
\noindent where $\mathbf{a} \triangleq [a_1, a_2, \ldots, a_N]^\text{T}$ denotes the activity status vector. 

Our goal is to minimize the negative log-likelihood function \( -\log p( \{\mathbf{y}_m\}_{m=1}^M ; \mathbf{a}) \), which is formulated as
\begin{align}
 \label{eq:optimization_problem}
\underset{\mathbf{a}\in [0, 1]^{N}}{\text{min}}\sum_{m=1}^M \left\{\log|\mathbf{C}_{m}| \!+ \!(\mathbf{y}_m \!-\!\bar{\mathbf{y}}_{m})^\text{H}\mathbf{C}_{m}^{-1} (\mathbf{y}_m \!- \!\bar{\mathbf{y}}_{m})\right\}.
\end{align}
Once problem (\ref{eq:optimization_problem}) is solved, the device activity is detected as \(\hat{a}_{n} = \mathbb{I}(a_{n} \geq \gamma)\), where $\mathbb{I}(\cdot)$ is an indicator function, and \(\gamma\) serves as a threshold within the range \([0, 1]\) to balance the probability of missed detection (PM) and the probability of false alarm (PF)~\cite{ganesan2021clustering}.

\textcolor{black}{\emph{Remark 1 (Extension to Multi-Antenna Users):} When the devices are equipped with multiple antennas, the signature sequence of each device can be transmitted with the help of beamforming. Specifically, let $N_\text{t} \ge 1$ denote the number of antennas at each device. Let $\mathbf{H}_{m,n} \in \mathbb{C}^{K \times N_\text{t}}$ denote the channel matrix from device $n$ to AP $m$, where each column of $\mathbf{H}_{m,n}$ independently follows the same distribution as the single-antenna channel $\mathbf{h}_{m,n}$. Let $\mathbf{w}_n \in \mathbb{C}^{N_\text{t}}$ with $\|\mathbf{w}_n\|_2 = 1$ denote the given beamforming vector at device $n$. Consequently, the effective channel from device $n$ to AP $m$ can be written as $\mathbf{h}_{m,n}^{\text{eff}} = \mathbf{H}_{m,n} \mathbf{w}_n$. Since $\mathbf{h}_{m,n}^{\text{eff}}$ is a linear combination of independent complex Gaussian columns of $\mathbf{H}_{m,n}$, it also follows a complex Gaussian distribution. Specifically, for far-field devices, since each column of $\mathbf{H}_{m,n}$ follows $\mathcal{CN}(\mathbf{0}, g_{m,n}\mathbf{I}_K)$, we have
\begin{align*}
\mathbf{h}_{m,n}^{\text{eff}} \sim \mathcal{CN}(\mathbf{0},\, g_{m,n}\mathbf{I}_K),~\forall~n \in \mathcal{U}^{\text{c}}_m,
\end{align*}
which has the same distribution as the single-antenna channel. For near-field devices, since each column of $\mathbf{H}_{m,n}$ follows $\mathcal{CN}(\beta_{m,n} \mathbf{b}(\mathbf{r}_{m,n}),\, \mathbf{R}_{m,n})$, we have
\begin{align*}
\mathbf{h}_{m,n}^{\text{eff}} \sim \mathcal{CN}\!\left((\mathbf{1}_{N_\text{t}}^\text{T} \mathbf{w}_n)\, \beta_{m,n} \mathbf{b}(\mathbf{r}_{m,n}),\, \mathbf{R}_{m,n}\right),~\forall~n\in\mathcal{U}_m,
\end{align*}
where the covariance matrix remains the same as the single-antenna case, and the mean differs only by a known scalar factor $\mathbf{1}_{N_\text{t}}^\text{T} \mathbf{w}_n$. Consequently, with the effective channel $\mathbf{h}_{m,n}^{\text{eff}}$, the mean $\bar{\mathbf{y}}_{m}$ in~(8) becomes
\begin{align*}
\bar{\mathbf{y}}_{m} =  \sum_{n \in \mathcal{U}_m} a_n \left(\mathbf{1}_{N_\text{t}}^\text{T} \mathbf{w}_n\right) \beta_{m,n} \mathbf{b}(\mathbf{r}_{m,n}) \otimes \mathbf{s}_n,
\end{align*}
while the covariance matrix $\mathbf{C}_{m}$ in~(\ref{expectation and covariance}) remains unchanged. Therefore, the detection problem can still be formulated as problem~(\ref{eq:optimization_problem}), and hence can be solved by the subsequent proposed algorithms.}

\textcolor{black}{\emph{Remark 2 (LoS or Correlated Far-Field Channel):} The Rayleigh fading model (\ref{far-field chanel}) is adopted here because far-field devices are typically at a greater distance from the AP, where the LoS path is more likely to be obstructed. When the far-field channels also contain a LoS component, e.g., $\mathbf{h}_{m,n} \sim \mathcal{CN}(\bar{\mathbf{h}}_{m,n},\, g_{m,n}\mathbf{I}_K)$, the mean $\bar{\mathbf{y}}_m$ in~(\ref{expectation and covariance}) would include an additional term $\sum_{n \in \mathcal{U}_m^{\text{c}}} a_n \bar{\mathbf{h}}_{m,n} \otimes \mathbf{s}_n$, while the covariance matrix $\mathbf{C}_m$ retains the same, and thus the proposed formulation still remains applicable. Futhermore, when the far-field channels are spatially correlated, i.e., $\mathbf{h}_{m,n} \sim \mathcal{CN}(\mathbf{0}, \mathbf{R}_{m,n}^{\text{far}})$ with a general covariance matrix $\mathbf{R}_{m,n}^{\text{far}}$, the proposed formulation and algorithm remain applicable by replacing $g_{m,n}\mathbf{I}_K$ with $\mathbf{R}_{m,n}^{\text{far}}$ in the expression of $\mathbf{C}_m$ in~(\ref{expectation and covariance}). The extension is also valid even when the far-field channel contains both LoS component and is spatially correlated.  }

\section{Impact of Hybrid Near-Far Field Channels on Detection Performance}
\label{sec:detection_performance}
Before solving problem (\ref{eq:optimization_problem}), we first analyze how the hybrid near-far field channels affect the detection performance  based on the established joint PDF (\ref{joint PDF}). For notational simplicity, we  introduce the following unified symbol. Let 
\begin{align}
\mathbf{\Xi}_{m,n} &=
\begin{cases} 
\mathbf{R}_{m,n}, & \text{if } n \in \mathcal{U}_m, \\
g_{m,n}\mathbf{I}_K, & \text{if } n \in \mathcal{U}_m^{\text{c}}.
\end{cases}
\end{align}
Then, we can define a matrix $\mathbf{X}_{m,n}=\mathbf{\Xi}_{m,n}^{\frac{1}{2}} \otimes \mathbf{s}_n$ with
\begin{align}
\label{defineX_mn}
 \mathbf{X}_{m,n} \mathbf{X}_{m,n}^\text{H} = \mathbf{\Xi}_{m,n} \otimes \left( \mathbf{s}_n \mathbf{s}_n^\text{H} \right).
\end{align}

To achieve good detection performance, the true activity indicator vector $\mathbf{a}^{\circ}$ should be uniquely identifiable. Specifically, there should not exist another vector $\tilde{\mathbf{a}} \neq \mathbf{a}^{\circ}$ that yields  $p(\{\mathbf{y}_m\}_{m=1}^M ; \tilde{\mathbf{a}}) = p(\{\mathbf{y}_m\}_{m=1}^M ; \mathbf{a}^{\circ}).$ Since the joint PDF (\ref{joint PDF}) is multivariate Gaussian, this equivalence requires identical means and covariance matrices.  According to the second equation of (\ref{expectation and covariance}), when the covariance matrices are identical, we have 
\begin{equation}
	\label{eq:identifiability}
    \sum_{n=1}^N \!(\tilde{a}_n\!-\!a_n^{\circ}) \mathbf{X}_{m,n}\mathbf{X}_{m,n}^\text{H}\!=\!\mathbf{O},
\end{equation}
where $\tilde{a}_n$ and $a_n^{\circ}$ denote the $n$-th entry of $\tilde{\mathbf{a}}$ and $\mathbf{a}^{\circ}$, respectively, and  $\mathbf{O}$ denotes the zero matrix. 

Letting $\boldsymbol{\xi} =[\xi_1, \xi_2, \dots, \xi_N]^\text{T}=\tilde{\mathbf{a}} - \mathbf{a}^{\circ}$, $\boldsymbol{\psi}_{m,n}= \text{vec}(\mathbf{X}_{m,n}\mathbf{X}_{m,n}^\text{H}\!)$, and $\boldsymbol{\Psi}_m =[\boldsymbol{\psi}_{m,1}, \boldsymbol{\psi}_{m,2},\dots, \boldsymbol{\psi}_{m,N}]$, (\ref{eq:identifiability}) becomes
\begin{equation}
	\label{eq:identifiability_matrix}
 \boldsymbol{\Psi}_m \boldsymbol{\xi}  = \mathbf{0}.
\end{equation}
 Let
\begin{align}
    \mathcal{V} &= \{\boldsymbol{\xi} \in \mathbb{R}^N | \boldsymbol{\Psi}_m\boldsymbol{\xi} = \mathbf{0}, \boldsymbol{\xi} \neq \mathbf{0}, \forall~m =1,2,...,M\}, \label{identify 1}\\
    \mathcal{D} &= \{\boldsymbol{\xi} \in \mathbb{R}^N | \xi_n \geq 0 \text{ if } a_n^{\circ} = 0, \xi_n \leq 0 \text{ if } a_n^{\circ} = 1\}\label{feasible set},
\end{align}
where  (\ref{identify 1}) defines the solution set containing all $\tilde{\mathbf{a}}\neq \mathbf{a}^{\circ}$ that yields (\ref{eq:identifiability_matrix}), while (\ref{feasible set}) defines the feasible set of problem (\ref{eq:optimization_problem}). Consequently, the identifiability of $\mathbf{a}^{\circ}$  requires  \(\mathcal{V}\cap \mathcal{D} = \emptyset \). From (\ref{identify 1}), we can observe that  if the columns of  $\boldsymbol{\Psi}_m$ are closer to orthogonality, then the condition $\mathcal{V} \cap \mathcal{D} = \emptyset$ is more likely to hold, which ensures the identifiability of $\mathbf{a}^{\circ}$~\cite{jianwu2008}. Then, we further analyze the orthogonality among columns of $\boldsymbol{\Psi}_m$. Specifically, we define the cosine similarity between any two columns $\boldsymbol{\psi}_{m,n}$ and $\boldsymbol{\psi}_{m,n'}$ of $\boldsymbol{\Psi}_m$ as
\begin{align}
	\label{pho similarity}
\rho_{n,n'} = \frac{\boldsymbol{\psi}_{m,n}^{\text{H}} \boldsymbol{\psi}_{m,n'}}{\|\boldsymbol{\psi}_{m,n}\|_2 \|\boldsymbol{\psi}_{m,n'}\|_2}.
\end{align}

The following proposition demonstrates that the hybrid near-far field channels enhance the orthogonality among columns of $\boldsymbol{\Psi}_m$, thereby improving the detection performance.

\begin{proposition}\label{prop:mutual_similarity}
Consider any two devices $n$ and $n'$ in the system. The cosine similarity between the corresponding columns $\boldsymbol{\psi}_{m,n}$ and $\boldsymbol{\psi}_{m,n'}$ of $\boldsymbol{\Psi}_m$ satisfies the following ordering based on the channel types:
\begin{align}
	\rho_{\text{NF-NF}} \leq \rho_{\text{FF-FF}},\,\, \rho_{\text{NF-FF}} < \rho_{\text{FF-FF}},
\end{align}
where $\rho_{\text{NF-FF}}$, $\rho_{\text{NF-NF}}$, and $\rho_{\text{FF-FF}}$ denote the cosine similarities for near-field to far-field, near-field to near-field, and far-field to far-field device pairs, respectively. 
\end{proposition}

\begin{proof}
	Please see Appendix~\ref{app1}.
\end{proof}
Based on Proposition~\ref{prop:mutual_similarity}, we summarize the impact of hybrid near-far field channels in Fig.~\ref{fig:impact_of_hybrid_near_far_field_channels}.

\begin{figure}[t!]
    \centering
    \includegraphics[width=1\linewidth]{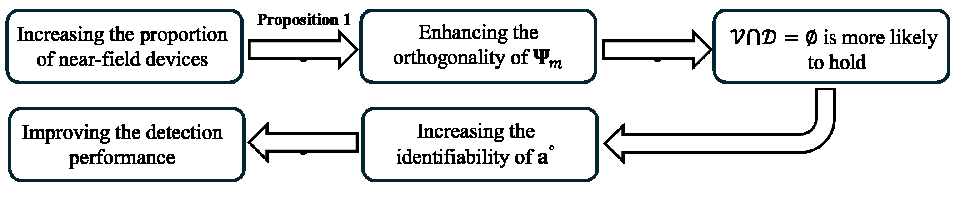}
    \caption{\textcolor{black}{Illustration of how increasing the proportion of near-field devices improves the detection performance through the enhanced orthogonality.}}
    \label{fig:impact_of_hybrid_near_far_field_channels}
\end{figure}

\section{A Unified Distributed  Activity Detection Algorithm}
\label{sec:distributed}
In this section, we develop a unified distributed algorithm to tackle problem (\ref{eq:optimization_problem}), which includes the activity detection for the single-cell near-field or far-field, and the cell-free near-field or far-field communications as special cases. 

\subsection{Distributed Algorithm Framework}
Unlike centralized method where the CPU handles all computations, the proposed distributed algorithm operates at both the APs and the CPU. Each AP conducts their local detection and then transmits the detection result to the CPU. Then, the CPU  aggregates these results and distributes the combined information back to all APs to prepare for the subsequent iteration.

To enable distributed processing, we first reformulate problem (\ref{eq:optimization_problem}) as an equivalent consensus form:
\begin{subequations}
\label{distributed formulation}
\begin{align}
&\min_{\left \{{\boldsymbol{\theta}_{m}}\right \}_{m=1}^{M}, \mathbf {a}\in \left [{0,1}\right]^{N}}\, 
\sum _{m=1}^{M}f_{m}\left ({\boldsymbol{\theta}_{m}}\right)
\\&\qquad\quad ~\,\text {s.t.}~\qquad \boldsymbol{\theta}_{m}=\mathbf {a},\quad \forall~m=1,2,\ldots,M,
\end{align}    
\end{subequations}
where  $f_{m}\left ({\boldsymbol{\theta}_{m}}\right)= \log|\tilde{\mathbf{C}}_{m}| + (\mathbf{y}_m - \tilde{\bar{\mathbf{y}}}_{m})^\text{H} \tilde{\mathbf{C}}_{m}^{-1} (\mathbf{y}_m - \tilde{\bar{\mathbf{y}}}_{m})$, and $\tilde{\mathbf{C}}_{m}$ and $ \tilde{\bar{\mathbf{y}}}_{m}$ follow the same structure as $\mathbf{C}_{m}$ and $\bar{\mathbf{y}}_{m}$ but with $\boldsymbol{\theta}_{m}$ replacing \(\mathbf{a}\) in (\ref{expectation and covariance}). Note that $f_m(\cdot)$ only depends on local parameters at the $m$-th AP, making it suitable for distributed optimization.

To solve problem (\ref{distributed formulation}) in a distributed manner, we construct its augmented Lagrangian function:
\begin{align}
\label{lagrangian}
&\mathcal{L} \left({\left \{{\boldsymbol{\theta}_{m}}\right \}_{m=1}^{M}, \mathbf {a}; \left \{{\boldsymbol {\lambda }_{m}}\right \}_{m=1}^{M} }\right) \nonumber \\
=&\sum \limits _{m=1}^{M} \left\{f_{m}\left ({\boldsymbol{\theta}_{m}}\right)+\boldsymbol {\lambda }_{m}^{\mathrm {T}}\left ({\boldsymbol{\theta}_{m}-\mathbf {a}}\right) +\frac {\mu }{2}\left \Vert{ \boldsymbol{\theta}_{m}-\mathbf {a}}\right \Vert_{2}^{2}\right\},
\end{align}
where $\boldsymbol{\lambda}_m \in \mathbb{R}^{N}$ is the dual variable associated with the consensus constraint $\boldsymbol{\theta}_{m} = \mathbf{a}$, and $\mu > 0$ is a penalty parameter. We adopt the alternating direction method of multipliers framework by alternately updating $\mathbf{a}$, $\{\boldsymbol{\theta}_{m}\}_{m=1}^M$, and $\{\boldsymbol{\lambda}_m\}_{m=1}^M$.

\subsubsection{Subproblem with respect to $\left \{{\boldsymbol{\theta}_{m}}\right \}_{m=1}^{M}$} At the $i$-th iteration, we decompose the problem with respect to $\left \{{\boldsymbol{\theta}_{m}}\right \}_{m=1}^{M}$ into $M$ parallel \mbox{subproblems}, with each  processed by the corresponding AP and given by
\begin{align}
\label{sub b}
\min_{\boldsymbol{\theta}_{m}\in [0,1]^N}f_{m}\left (\!{\boldsymbol{\theta}_{m}}\!\right)\!+\!\!\left ({\boldsymbol {\lambda }_{m}^{(i-1)}}\right)^{\mathrm {T}} \left ({\boldsymbol{\theta}_{m}\!\!-\!\!\mathbf {a}^{(i-1)}}\right)\! +\!\frac {\mu }{2}\left \Vert{ \boldsymbol{\theta}_{m}\!\!-\!\!\mathbf {a}^{(i-1)}}\right \Vert_{2}^{2}. 
\end{align}
This is a single-cell  activity detection
problem with additional
linear and quadratic terms. However, due to  channel correlation, the rank-one update in the existing covariance-based approach \cite{ganesan2021clustering,li2022asynchronous,wang2024scalinglaw} is not applicable anymore. To address this challenge, we propose a CD algorithm to solve problem (\ref{sub b}), which will be discussed in detail in Section \ref{CD Algorithm}.

\subsubsection{Updating the dual variables $\{\boldsymbol{\lambda}_{m}\}_{m=1}^{M}$} After updating  $\{\boldsymbol{\theta}_{m}\}_{m=1}^M$, $\{\boldsymbol{\lambda}_{m}\}_{m=1}^{M}$ are updated by a dual ascent step:
\begin{align}
\label{dual ascent step}
\boldsymbol{\lambda}_m^{(i)} \!= \!\boldsymbol{\lambda}_m^{(i-1)} + \mu \left( \boldsymbol{\theta}_{m}^{(i)} - \mathbf{a}^{(i-1)} \right), \forall~m = 1,2, \dots, M. 
\end{align}

\subsubsection{Subproblem with respect to $\mathbf{a}$}  We decompose the problem with respect to $\mathbf{a}$ into $N$ parallel subproblems, with each written as
\begin{align}
\label{sub a}
\min \limits _{a_n \in \left [{0,1}\right]} \sum _{m=1}^{M} \left\{{\lambda_{m,n}^{(i)}}\left ({\theta_{m,n}^{(i)}\!-\!a_n}\right)+\frac {\mu }{2}\left ({\theta_{m,n}^{(i)}-a_n}\right)^{2}\right\},
\end{align}
where  $\theta_{m,n}^{(i)}$ denotes the $n$-th entry of $\boldsymbol{\theta}_{m}^{(i)}$. Problem (\ref{sub a}) is a one-dimensional convex quadratic problem, and hence its optimal solution can be derived in a closed form:
\begin{align}
a_{n}^{(i)} = \Pi_{[0,1]} \left( \delta_{n}^{(i)} \right), 
\quad \forall~n = 1,2, \dots, N,
\label{eq:close form}
\end{align}
where $\delta_{n}^{(i)} = \sum_{m=1}^M \left( \mu \theta_{m,n}^{(i)} + \lambda_{m,n}^{(i)}\right)/(M\mu)$, and $\Pi_{[0,1]} (\cdot)$ is the projection operation onto $[0,1]$.

Through iterative updates of both primal and dual variables,  the distributed optimization framework is outlined in Algorithm~\ref{distributed algorithm}.

\begin{algorithm}[t]
	\caption{Proposed Distributed Algorithm for Solving Problem (\ref{distributed formulation})}
	\label{distributed algorithm}
	\begin{algorithmic}[1]
	\STATE \textbf{Initialize:}  $ \boldsymbol{\theta}_{m}^{(0)}$, $\boldsymbol{\lambda}_{m}^{(0)}$,  $\forall~m = 1, 2, \dots, M$, and $\mathbf{a}^{(0)}$ using~(\ref{eq:close form});
	\STATE \textbf{repeat} ($i = 1, 2,\dots$)
	\STATE \quad The CPU broadcasts $\mathbf{a}^{(i-1)}$ to each AP $m$, $\forall~m = 1, 2,\dots, M$;
	\STATE  \quad Each AP $m$ updates $\boldsymbol{\theta}_{m}^{(i)}$ by solving~(\ref{sub b}), $\forall~m = 1,2,\dots, M$;
	\STATE \quad Each AP $m$ updates $\boldsymbol{\lambda}_m^{(i)}$ by a dual ascent step~(\ref{dual ascent step}), $\forall~m = 1,2, \dots, M$;
	\STATE \quad Each AP $m$ sends $\mu \boldsymbol{\theta}_{m}^{(i)} + \boldsymbol{\lambda}_m^{(i)}$ to the CPU, $\forall~m = 1,2, \dots, M$;
	\STATE \quad The CPU updates $\mathbf{a}^{(i)}$ with (\ref{eq:close form});
	\STATE \textbf{until} convergence
	\end{algorithmic}
	\end{algorithm}

\subsection{Proposed CD Algorithm for Solving Subproblem (\ref{sub b})}
\label{CD Algorithm}
Due to the inapplicability of the conventional rank-1 update, 
we propose a novel CD algorithm to solve subproblem~(\ref{sub b}). Specifically, we randomly permute the indices of all coordinates $\{1,2,...,N\}$, and update each coordinate $\theta_{m,n}$ sequentially according to this permutation. The update of $\theta_{m,n}$ can be expressed as $\theta_{m,n}+d$ with $\{\theta_{m,n'}\}_{n'=1,n'\neq n}^N$ fixed, where $d$ can be determined by solving a one-dimensional subproblem. For notational simplicity, we define
\begin{align}
\boldsymbol{\varpi}_{m,n} &= 
\begin{cases} 
\beta_{m,n} \mathbf{b}(\mathbf{r}_{m,n}), & \text{if } n \in \mathcal{U}_m, \\
\mathbf{0}, & \text{if } n \in \mathcal{U}_m^{\text{c}}.
\end{cases} 
\end{align}
Let $\mathbf{e}_n \in \mathbb{R}^N$ denote the standard basis vector whose $n$-th entry is 1 while all other entries are 0. Consequently, treating $d$ as the optimization variable,  substituting the expression of $f_m(\boldsymbol{\theta}_m)$ into (\ref{sub b}), and using the definition of $\mathbf{X}_{m,n}$ in (\ref{defineX_mn}),  (\ref{sub b}) can be written as
\begin{align}
\label{origin_update_distributed}
  & \underset{d \in [-\theta_{m,n},\,1-\theta_{m,n}]}{\operatorname{min}} \log \left| \tilde{\mathbf{C}}_{m} + d \mathbf{X}_{m,n} \mathbf{X}_{m,n}^{\text{H}} \right|  \notag\\
    &\quad\quad+ \left( \mathbf{y}_m - \tilde{\bar{\mathbf{y}}}_{m} - d\boldsymbol{\varpi}_{m,n} \otimes \mathbf{s}_n \right)^{\text{H}} \left( \tilde{\mathbf{C}}_{m} + d \mathbf{X}_{m,n} \mathbf{X}_{m,n}^{\text{H}} \right)^{-1}\notag\\
    &\quad\quad\times  \left( \mathbf{y}_m - \tilde{\bar{\mathbf{y}}}_{m} - d\boldsymbol{\varpi}_{m,n} \otimes \mathbf{s}_n \right) \notag\\
    &\quad\quad+ (\boldsymbol{\lambda}_m^{(i-1)})^{\text{T}}(\boldsymbol{\theta}_{m} \!\!+\! d\mathbf{e}_n \!\!-\! \mathbf{a}^{(i-1)})
   \! +\! \frac{\mu}{2}\|\boldsymbol{\theta}_{m}\! + \!d\mathbf{e}_n \!\!- \!\mathbf{a}^{(i-1)}\|_2^2. \notag\\
\end{align}

\textcolor{black}{Compared with the traditional far-field only scenario where $\tilde{\mathbf{C}}_{m}$ reduces to $K$ identical $L \times L$ blocks with each block $\mathbf{Q}_m=\sum_{n=1}^{N} \theta_{m,n} g_{m,n} \mathbf{s}_n \mathbf{s}_n^\text{H}+\varsigma^2_m \mathbf{I}_{L}$~\cite{chen2021phase,lin2024communication,ganesan2021clustering}, the dimension of the covariance matrix $\tilde{\mathbf{C}}_{m}$ in problem (\ref{origin_update_distributed}) is $LK \times LK$. Furthermore, the block diagonal structure of $\tilde{\mathbf{C}}_{m}$ in far-field scenario  allows the rank-1 update $(\mathbf{Q}_{m} + d \cdot g_{m,n} \mathbf{s}_n \mathbf{s}_n^\text{H})^{-1} = \mathbf{Q}_{m}^{-1} - d \cdot g_{m,n} \frac{\mathbf{Q}_{m}^{-1} \mathbf{s}_n \mathbf{s}_n^\text{H} \mathbf{Q}_{m}^{-1}}{1 + d \cdot g_{m,n} \mathbf{s}_n^\text{H} \mathbf{Q}_{m}^{-1} \mathbf{s}_n}$. However, when near-field devices exist, the channel correlation
destroys the block-diagonal structure of $\tilde{\mathbf{C}}_{m}$ and makes the
perturbation no longer decomposable into independent rank-one
updates per block. Hence, updating $\tilde{\mathbf{C}}_{m}$  in the hybrid-field case is much more complicated than the traditional far-field setting.} To address this challenge, we rewrite the log-determinant term as
\begin{align}\label{eq:term1}
	& \log\left| \tilde{\mathbf{C}}_{m} + d \, \mathbf{X}_{m,n} \mathbf{X}_{m,n} ^\text{H}\right| \nonumber \\
	= &\log \left| \mathbf{I}_{LK} + d\,\mathbf{X}_{m,n} \mathbf{X}_{m,n} ^\text{H} \tilde{\mathbf{C}}_{m}^{-1} \right| + \log \left| \tilde{\mathbf{C}}_{m}\right| \nonumber \\
	= &\log \left| \mathbf{I}_{J_{m,n}} + d\, \mathbf{X}_{m,n} ^\text{H} \tilde{\mathbf{C}}_{m}^{-1} \mathbf{X}_{m,n} \right| + \log \left| \tilde{\mathbf{C}}_{m}\right|\nonumber\\
    \overset{(a)}{\approx} &d\operatorname{tr} \left( \mathbf{X}_{m,n} ^\text{H} \tilde{\mathbf{C}}_{m}^{-1} \mathbf{X}_{m,n} \right) + \log \left| \tilde{\mathbf{C}}_{m}\right|,
\end{align}
where $J_{m,n}= \operatorname{rank}(\mathbf{\Xi}_{m,n})$. The first two equalities in (\ref{eq:term1}) are derived from the property of the determinant to reduce the matrix dimension, and $(a)$ is derived from the first-order Taylor expansion to avoid the computational complexity in calculating the determinant. On the other hand, for the matrix inverse term, based on the Sherman-Morrison-Woodbury formula, we have
\begin{align}\label{eq:term2}
	&\left(\tilde{\mathbf{C}}_{m} + d \mathbf{X}_{m,n} \mathbf{X}_{m,n}^{\text{H}} \right)^{-1} \notag\\
    =&\,\,\,\tilde{\mathbf{C}}_{m}^{-1}- d\, \tilde{\mathbf{C}}_{m}^{-1} \mathbf{X}_{m,n}\notag\\
    &\quad\quad\,\,\,\times\left( \mathbf{I}_{J_{m,n}} + d\, \mathbf{X}_{m,n}^{\text{H}} \tilde{\mathbf{C}}_{m}^{-1} \mathbf{X}_{m,n} \right)^{-1} \mathbf{X}_{m,n}^{\text{H}} \tilde{\mathbf{C}}_{m}^{-1}\notag\\
	\overset{(b)}{\approx}&\,\,\,\tilde{\mathbf{C}}_{m}^{-1}- d\, \tilde{\mathbf{C}}_{m}^{-1} \mathbf{X}_{m,n} \left( \mathbf{I}_{J_{m,n}} - d\, \mathbf{X}_{m,n}^{\text{H}} \tilde{\mathbf{C}}_{m}^{-1} \mathbf{X}_{m,n} \right) \notag\\
    &\quad\quad\,\,\,\times\mathbf{X}_{m,n}^{\text{H}} \tilde{\mathbf{C}}_{m}^{-1},
		\end{align}
where $(b)$ also applies the first-order Taylor expansion to avoid the computational complexity in repeatedly calculating the matrix inverse. \textcolor{black}{Substituting the approximations~(\ref{eq:term1}) and~(\ref{eq:term2}) into subproblem~(\ref{origin_update_distributed}) and adding a penalty term $(\omega/2)d^2$ with a properly chosen $\omega$, (\ref{origin_update_distributed}) is transformed into:}
\begin{equation}\label{eq:one-dim-approx}
	\underset{d \in [-\theta_{m,n},\,1-\theta_{m,n}]}{\operatorname{min}} \quad p_{\mathrm{approx}}(d) + \frac{\omega}{2} d^2,
\end{equation}
where $p_{\mathrm{approx}}(d)$ is given by (\ref{eq:distributed inexact}) in Appendix~\ref{appB} with $\lambda_{m,n}^{(i-1)}$ denoting the $n$-th entry of $\boldsymbol{\lambda}_m^{(i-1)}$. \textcolor{black}{As will be shown in Proposition~\ref{proposition1} in Section~\ref{convergence analysis}, the penalty term $(\omega/2)d^2$ compensates for the approximation error in~(\ref{eq:term1}) and~(\ref{eq:term2}), ensuring that each coordinate update sufficiently decreases the original objective function of problem~(\ref{sub b}), and thereby guaranteeing convergence to a stationary point.} Setting the gradient of  (\ref{eq:one-dim-approx}) to zero,  we could obtain three roots. From the roots and the two boundary points
$-\theta_{m,n}$ and $1 -\theta_{m,n}$, the  optimal solution can be obtained by selecting the one corresponding to the smallest objective value. After sequentially updating $\theta_{m,n}$ for all $n$, we obtain $\boldsymbol{\theta}^{(i)}_m$ at each AP $m$, and the above  procedure is summarized as Algorithm~\ref{alg:cd}.

\begin{algorithm}[t]
	\caption{Proposed CD Algorithm for Solving Subproblem~(\ref{sub b})}%
	\label{alg:cd}
	\begin{algorithmic}[1]
		\STATE \textbf{Input:} $\mathbf{a}^{(i-1)}$, $\boldsymbol{\lambda}_m^{(i-1)}$, $\mathbf{y}_m$;
		\STATE \textbf{Initialize:} $\boldsymbol{\theta}_m \leftarrow  \mathbf{a}^{(i-1)}$;
		\STATE \textbf{repeat} 
		\STATE  \quad Randomly select a permutation $\{ q_1, q_2, \ldots, q_{N} \}$ of the coordinate indices $\{1, 2, \ldots, N\}$;
		\STATE \quad \textbf{for} {$n = q_1,$ $q_2, \ldots, q_{N}$} \textbf{do}
		\STATE  \quad\quad Solve problem (\ref{eq:one-dim-approx}) using the cubic formula to obtain~$\bar{d}$ for coordinate $n$;
		\STATE  \quad\quad $\boldsymbol{\theta}_m \leftarrow \boldsymbol{\theta}_m + \bar{d} \mathbf{e}_n$;
		\STATE  \quad \textbf{end for}
		\STATE \textbf{until} convergence
		\STATE \textbf{Output:} $\boldsymbol{\theta}_m^{(i)}\leftarrow \boldsymbol{\theta}_m$.
	\end{algorithmic}
\end{algorithm}

\subsection{Applicability to Special Cases}
	The proposed Algorithms \ref{distributed algorithm} and \ref{alg:cd} exhibit generalizability and can be applied to different conventional scenarios:
	\begin{itemize}
	\item \textbf{Cell-free case:} When all devices are located in the near-field regions ($\mathcal{U}^c_m = \emptyset$,~$\forall m = 1, 2, \ldots, M$), the proposed algorithms can be directly applied to the cell-free near-field activity detection. On the other hand, when all devices are located in the far-field regions ($\mathcal{U}_m = \emptyset$,~$\forall m = 1, 2, \ldots, M$), the covariance matrix $\tilde{\mathbf{C}}_m$ reduces to a block-diagonal structure:
	\begin{align}
	\tilde{\mathbf{C}}_m = \text{diag}\left( \mathbf{Q}_{m}, \mathbf{Q}_{m}, \ldots, \mathbf{Q}_{m} \right),
	\end{align}
	where $\mathbf{Q}_{m} = \sum_{n=1}^{N} \theta_{m,n} g_{m,n} \mathbf{s}_n \mathbf{s}_n^\text{H} + \varsigma^2_m \mathbf{I}_{L}$. In this case, (\ref{eq:term2}) reduces to the popular rank-1 update for each~$\mathbf{Q}_{m}$~\cite{li2022asynchronous,ganesan2021clustering}. Together with the update of $\boldsymbol{\theta}_m$ in line~7 of Algorithm~\ref{alg:cd}, the overall Algorithm~\ref{distributed algorithm} provides a distributed algorithm for the conventional cell-free far-field activity detection.  
	
	\item \textbf{Single-cell case:} When there exists only $M=1$ AP, Algorithm~\ref{alg:cd} can be individually applied for the single-cell activity detection. When all devices are located in the near-field region ($\mathcal{U}^c_1 = \emptyset$), it reduces to the single-cell near-field activity detection; when all devices are in the far-field region ($\mathcal{U}_1 = \emptyset$), it reduces to the single-cell far-field activity detection.
	\end{itemize}
	This unified algorithmic framework demonstrates broad applicability across diverse communication scenarios.

\subsection{Computational Complexity and Communication Overhead Analysis}
\label{sec:comparison}
In this subsection, we analyze the advantages of the proposed distributed algorithm over the centralized approach from two aspects: computational complexity and communication overhead. The analysis demonstrates that the proposed distributed algorithm can significantly reduce both of the computational burden at the CPU and the fronthaul overhead.

\subsubsection{Computational Complexity Analysis}

In Algorithm~\ref{distributed algorithm}, the dominant computational complexity at each AP $m$ lies in updating $\boldsymbol{\theta}_m$ by solving problem (\ref{sub b}). For each coordinate update, the computational cost consists of two main parts: (i) solving the problem (\ref{eq:one-dim-approx}) using the cubic formula with complexity $\mathcal{O}(1)$, and (ii) updating the matrix inverse $\tilde{\mathbf{C}}_{m}^{-1}$ using the Sherman-Morrison-Woodbury formula~(\ref{eq:term2}) with complexity $\mathcal{O}((LK)^2J_{m,n} + LKJ_{m,n}^2+J_{m,n}^3)$. Since $J_{m,n}= \operatorname{rank}(\mathbf{\Xi}_{m,n})\le K$, the overall complexity for each coordinate update is dominanted by $\mathcal{O}((LK)^2J_{m,n})$.

For comparison, if we solve problem (\ref{eq:optimization_problem}) in a centralized manner, it is equivalent to solving the following problem at each coordinate update using the CD algorithm:
\begin{align}
	\label{eq:centralized_obj}
	\underset{d \in [-a_{n},\,1-a_{n}]}{\operatorname{min}} &\sum_{m=1}^M  \left\{ d \hat{\rho}_{1,m}(\mathbf{a}) + d^2 \hat{\rho}_{2,m}(\mathbf{a}) \right. \nonumber \\
	& \left. + d^3 \hat{\rho}_{3,m}(\mathbf{a}) + d^4 \hat{\rho}_{4,m}(\mathbf{a}) \right\} + \frac{\tilde{\omega}}{2} d^2,
	\end{align}
where  $\hat{\rho}_{i,m}(\mathbf{a})$ for $i = 2, 3, 4$ are given by (\ref{eq:centralized_obj_terms}) in Appendix~\ref{appB}, and $\tilde{\omega} > 0$ is analogous to $\omega$ in (\ref{eq:one-dim-approx}).  Note that $\{\hat{\rho}_{i,m}(\mathbf{a})\}_{i=1}^4$ in (\ref{eq:centralized_obj_terms}) differs from $\{\rho_{i}(\boldsymbol{\theta}_{m})\}_{i=1}^4$ in (\ref{eq:distributed inexact})  as they do not contain the dual variables and the penalty related terms. For such a centralized approach, the summation over $M$ APs in (\ref{eq:centralized_obj}) leads to a complexity of $\mathcal{O}\!\left((LK)^2 \sum_{m=1}^{M} J_{m,n}\right)$ for each coordinate update.  \textcolor{black}{The computational complexity comparison is summarized in Table~\ref{tab:complexity_r2}.}  \textcolor{black}{We can see that the proposed distributed algorithm decomposes problem~(\ref{eq:optimization_problem}) so that each AP independently solves its own subproblem~(\ref{eq:one-dim-approx}) in parallel, yielding a per-coordinate complexity independent of $M$. In contrast, even with parallel computing at the CPU, problem~(\ref{eq:centralized_obj}) cannot be decomposed into $M$ independent subproblems due to the variable coupling among the $M$ components. Hence, its per-coordinate complexity inherently scales with $M$.} We will compare the computational complexity by simulations in Section \ref{sec:simulation}.

\begin{table}[t]
	\centering
	\caption{\textcolor{black}{Computational Complexity Comparison}}
	\label{tab:complexity_r2}
	\begin{tabular}{|c|c|c|}
	\hline
	& \textbf{Centralized} & \textbf{Proposed Distributed} \\
	\hline
	Per-coordinate & $\mathcal{O}\!\left((LK)^2 \sum_{m=1}^{M} J_{m,n}\right)$ & $\mathcal{O}((LK)^2 J_{m,n})$ \\
	\hline
	\end{tabular}
\end{table}

\subsubsection{Communication Overhead} 
In addition to computational complexity, the communication overhead between the APs and the CPU is another crucial aspect. For the proposed distributed algorithm, the signaling exchange occurs in two directions:
\begin{itemize}
    \item In the downlink, at each iteration $i$, the CPU broadcasts the $N$-dimensional vector $\mathbf{a}^{(i-1)}$ to each AP, requiring to transmit $MN$ real-valued numbers.
    \item In the uplink, each AP sends the $N$-dimensional vector $\mu \boldsymbol{\theta}_{m}^{(i)} + \boldsymbol{\lambda}_m^{(i)}$ to the CPU, requiring to transmit $MN$ real-valued numbers.
\end{itemize}
Therefore, with $I$ iterations, the total communication overhead of the proposed distributed algorithm is to transmit $2IMN$ real-valued numbers. Due to the small dynamic range $[0,1]$ of $\mathbf{a}$ and $\boldsymbol{\theta}_m$, the required number of quantization bits for transmitting the $2IMN$ numbers can be very small. Furthermore, since most of the entries in $\mathbf{a}$ and $\boldsymbol{\theta}_m$ are zeros due to the sparse activities, the communication overhead can be further reduced  by using various data compression schemes.

In contrast, for the centralized approach, each AP needs to transmit its received signal $\mathbf{Y}_m$ to the CPU with $2LK$ real-valued numbers\footnote{Different from the conventional far-field activity detection, where each AP $m$ can alternatively transmit $\mathbf{Y}_m\mathbf{Y}_m^{\text{H}}$ to the CPU, it should be noted from (\ref{eq:centralized_obj_terms}) that the terms such as $\left( \mathbf{y}_{m} - \bar{\mathbf{y}}_{m} \right)^\text{H} \mathbf{C}_{m}^{-1} \left( \boldsymbol{\varpi}_{m,n} \otimes \mathbf{s}_n \right)$ cannot be computed from $\mathbf{Y}_m\mathbf{Y}_m^{\text{H}}$ alone without the knowledge of $\mathbf{Y}_m$. Thus, for hybrid far-field activity detection, each AP $m$ must transmit $\mathbf{Y}_m$ to the CPU.}. Thus, the total real-valued numbers required to be transmitted by the centralized approach is $2MLK$. However, since the the dynamic range of $\mathbf{Y}_m$ can be much larger than that of $\mathbf{a}$ and $\boldsymbol{\theta}_m$, the required number of quantization bits for transmitting the $2MLK$ numbers can be much larger. We will further demonstrate the communication overhead comparison by simulations in Section \ref{sec:simulation}.

\section{Convergence Analysis} \label{convergence analysis}

Due to the approximation in (\ref{eq:one-dim-approx}), the convergence behavior of Algorithm~\ref{alg:cd} is unknown. To this end, the following proposition is presented on its convergence property.

\begin{proposition}
    \label{proposition1}
    Let $L_m > 0$ denote the Lipschitz constant of $\nabla U_m(\boldsymbol{\theta}_{m})$, where  $U_m(\boldsymbol{\theta}_m)$ denotes the objective function of problem (\ref{sub b}). There exists a constant $\rho=\bar{\rho}_2+\bar{\rho}_3+\bar{\rho}_4$ such that $\bar{\rho}_2$, $\bar{\rho}_3$, and $\bar{\rho}_4$ are the upper bounds of $\rho_{2}(\boldsymbol{\theta}_{m})$, $\rho_{3}(\boldsymbol{\theta}_{m})$, and $\rho_{4}(\boldsymbol{\theta}_{m})$ in (\ref{eq:distributed inexact}), respectively. Moreover, for any $\omega \geq L_m + \rho$ and $\omega+\rho \ge 1$, Algorithm~\ref{alg:cd} guarantees the convergence to at least a stationary point of problem (\ref{sub b}) .
    \end{proposition}

	\begin{proof}
		See Appendix \ref{app2}.
	\end{proof}
		
 Moreover, due to the non-convexity of problem (\ref{distributed formulation}) and the approximation in (\ref{eq:one-dim-approx}), it is still unknown about the overall convergence behavior of the proposed distributed algorithm. Based on Proposition~\ref{proposition1}, we can further establish the convergence behavior of the Algorithm~\ref{distributed algorithm}. The following theorem demonstrates that Algorithm~\ref{distributed algorithm} guarantees the convergence to a stationary point of problem~(\ref{distributed formulation}).

\begin{theorem}\label{Distributed Convergency Theorem}
When $\omega \ge L_{m} + \rho$, $\omega+\rho \ge 1$, and $\mu>2\tilde{L}_{m}$, where $\tilde{L}_{m}$ is the lipschitz constant of $\nabla f_m(\boldsymbol{\theta}_{m})$, the solution sequence  $\{\mathbf{a}^{(i)}\}$ generated by Algorithm \ref{distributed algorithm} (with $\boldsymbol{\theta}_n$ updated by Algorithm \ref{alg:cd}) converges to at least a stationary point of problem (\ref{distributed formulation}).
\end{theorem}

\begin{proof}
	See Appendix \ref{app3}.
\end{proof}

\section{Simulation Results}
\label{sec:simulation}
In this section, we validate the performance of the proposed method through simulations in terms of  PM and  PF~\cite{ganesan2021clustering}. \textcolor{black}{We consider a system with parameters summarized in Table~\ref{tab:simulation_parameters}.} APs and devices are uniformly distributed in the area. \textcolor{black}{The $L_m$ scatterers are independently and uniformly distributed within a disk of radius $20$~m centered at each AP~\cite{liuNearFieldCommunicationsTutorial2023}, \cite{Dong2022}. The channel gains $g_{m,n}$, $|\beta_{m,n}|^2$, and $|\tilde{\beta}_{m,n,\ell}|^2$ are computed from the path-loss model $128.1 + 37.6\log_{10}(\tau)$~(dB), where $\tau$ is the distance in kilometers between device~$n$ and AP~$m$ for $g_{m,n}$ and $|\beta_{m,n}|^2$, or the distance between device~$n$ and scatterer~$\ell$ of AP~$m$ for $|\tilde{\beta}_{m,n,\ell}|^2$ \cite{liuNearFieldCommunicationsTutorial2023}, \cite{Dong2022}, \cite{Wang2026}.} 

\begin{table}[t]
	\centering
	\caption{\textcolor{black}{Summary of Simulation Parameters}}
	\label{tab:simulation_parameters}
	\begin{tabular}{|c|c|}
	\hline
	\textbf{Parameter} & \textbf{Value} \\
	\hline
	Simulation area & $200 \times 200$~m$^2$ \\
	\hline
	Number of devices $N$ & $100$ (default), $[40,\, 120]$ \\
	\hline
	Active device ratio $\alpha$ & $0.1$ (default), $[0.05,\, 0.25]$ \\
	\hline
	Carrier wavelength $\lambda_\text{c}$ & $0.2$~m (default), $[0.05,\, 0.3]$~m \\
	\hline
	Signature sequence $ \mathbf{s}_n$ (uniformly sampled) & $\left\{ \pm \frac{\sqrt{2}}{2} \pm j \frac{\sqrt{2}}{2} \right\}^L$ \\
	\hline
	Number of scatterers per AP $L_m$ & $8$ (default), $[4,\, 20]$ \\
	\hline
	Reflection coefficient variance $\sigma_{m,\ell}^2$ & $1$ \\
	\hline
	Penalty parameter $\omega$ & $20$ \\
	\hline
	Background noise power & $-99$~dBm \\
	\hline
	\end{tabular}
\end{table}

\subsection{Convergence Behavior of Proposed Distributed Algorithm}

Figure~\ref{prob of error versus iteration number} demonstrates the convergence behavior of the proposed distributed algorithm  across different total antenna configurations, where the number of APs is set to $M=3$, the signature sequence length is $L=6$, and the total number of antennas $MK$ varies from $48$ to $96$ with the antennas distributed equally among the APs. We evaluate the detection performance by setting the threshold such that the PM equals to the PF. The results show that the distributed algorithm achieves fast convergence, reaching the performance of the centralized approach within only $2$ iterations.
This fast convergence can efficiently reduce the computational complexity and communication overhead of the proposed distributed algorithm, which will be further discussed in Section~\ref{sec:complexity_comparison}.

\begin{figure}[t]
	\centering	\includegraphics[width=0.85\columnwidth]{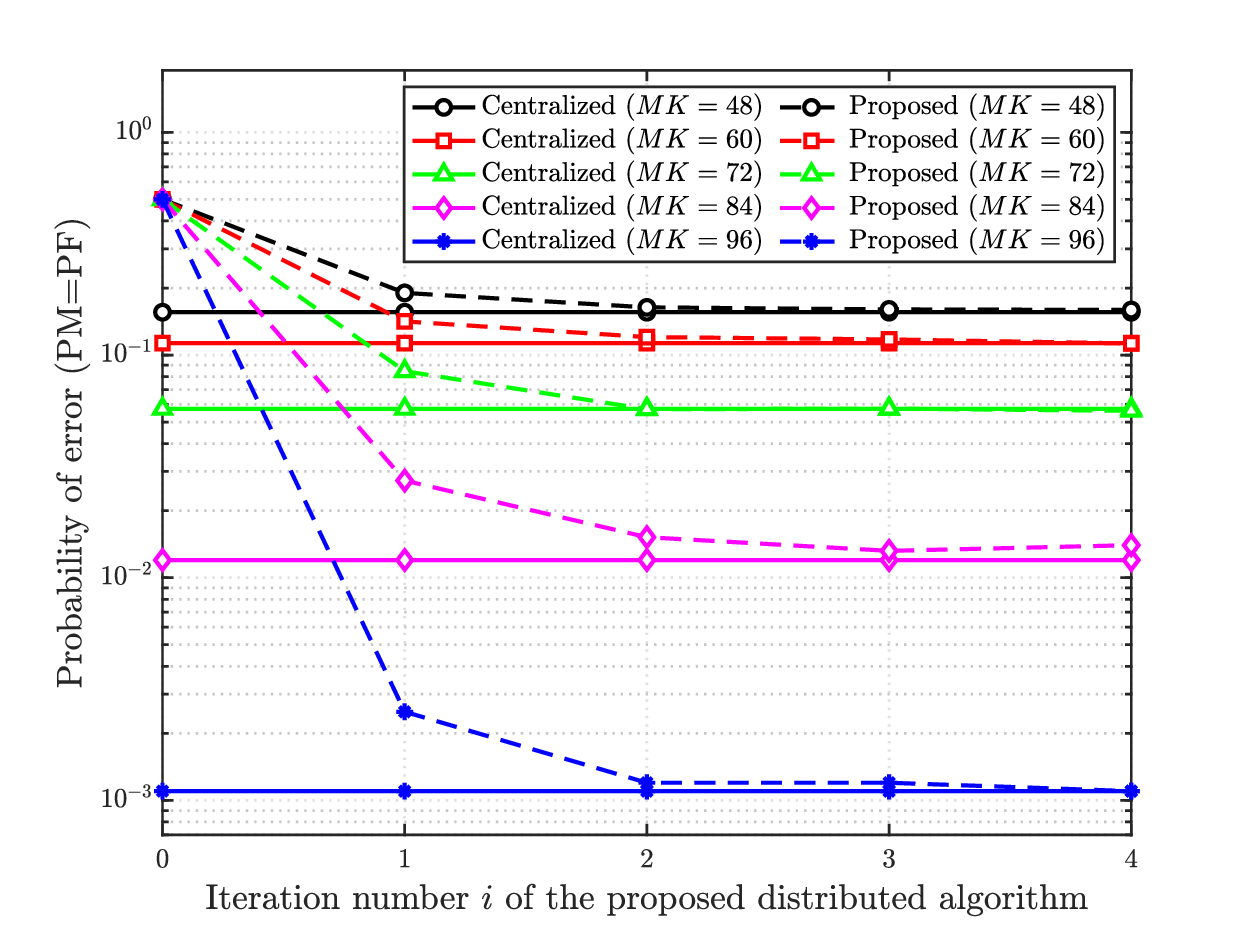}
	\caption{Probability of error versus iteration number.}
	\label{prob of error versus iteration number}
	\vspace{-10pt}
\end{figure}

Figure~\ref{prob of error versus number of APs for distributed algorithm} further illustrates the detection performance of the proposed distributed algorithm under different numbers of iterations, where the AP number varies from $1$ to $12$ with the total number of antennas fixed at $72$. It can be seen that the proposed distributed algorithm  rapidly converges within $1$ or $2$ iterations across different numbers of APs, while maintaining comparable detection performance to that of the centralized approach.

\begin{figure}[t]
	\centering
	\includegraphics[width=0.85\columnwidth]{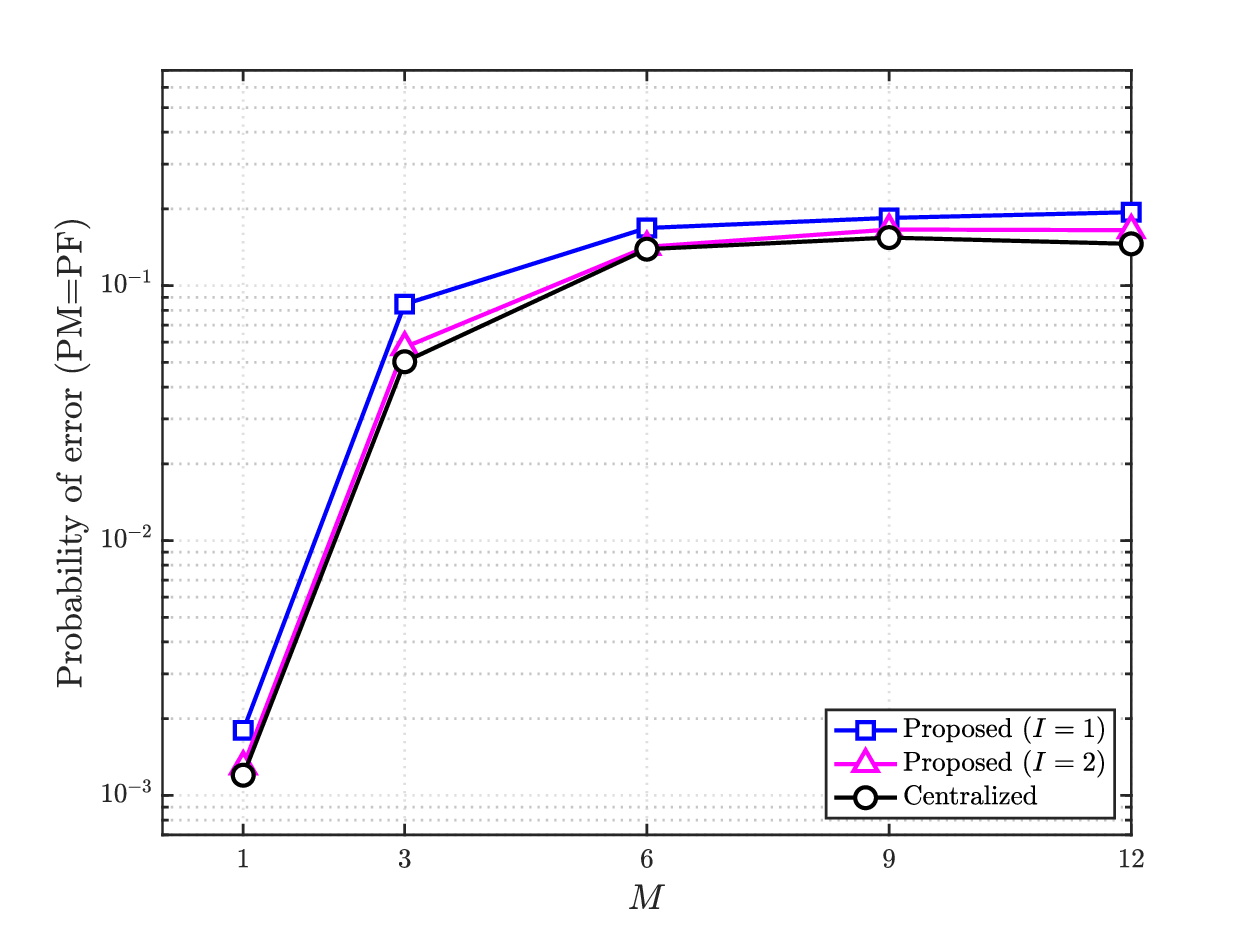}
	\caption{Probability of error versus number of APs.}
	\label{prob of error versus number of APs for distributed algorithm}
	\vspace{-10pt}
\end{figure}

\textcolor{black}{To further validate the convergence of Algorithm~\ref{alg:cd} as stated in Proposition~\ref{proposition1}, Fig.~\ref{fig:func_conv} plots the objective function value of problem (\ref{sub b}) versus the number of iterations of Algorithm~\ref{alg:cd} under different numbers of antennas, where $M=3$ and $L=6$. It can be observed that the objective function value  of problem (\ref{sub b}) decreases monotonically and converges within a few iterations, confirming that the Taylor approximation  does not hinder the convergence of Algorithm~\ref{alg:cd} for solving problem~(\ref{sub b}) as stated in Proposition~\ref{proposition1}.}

\begin{figure}[t]
	\centering
	\includegraphics[width=0.85\columnwidth]{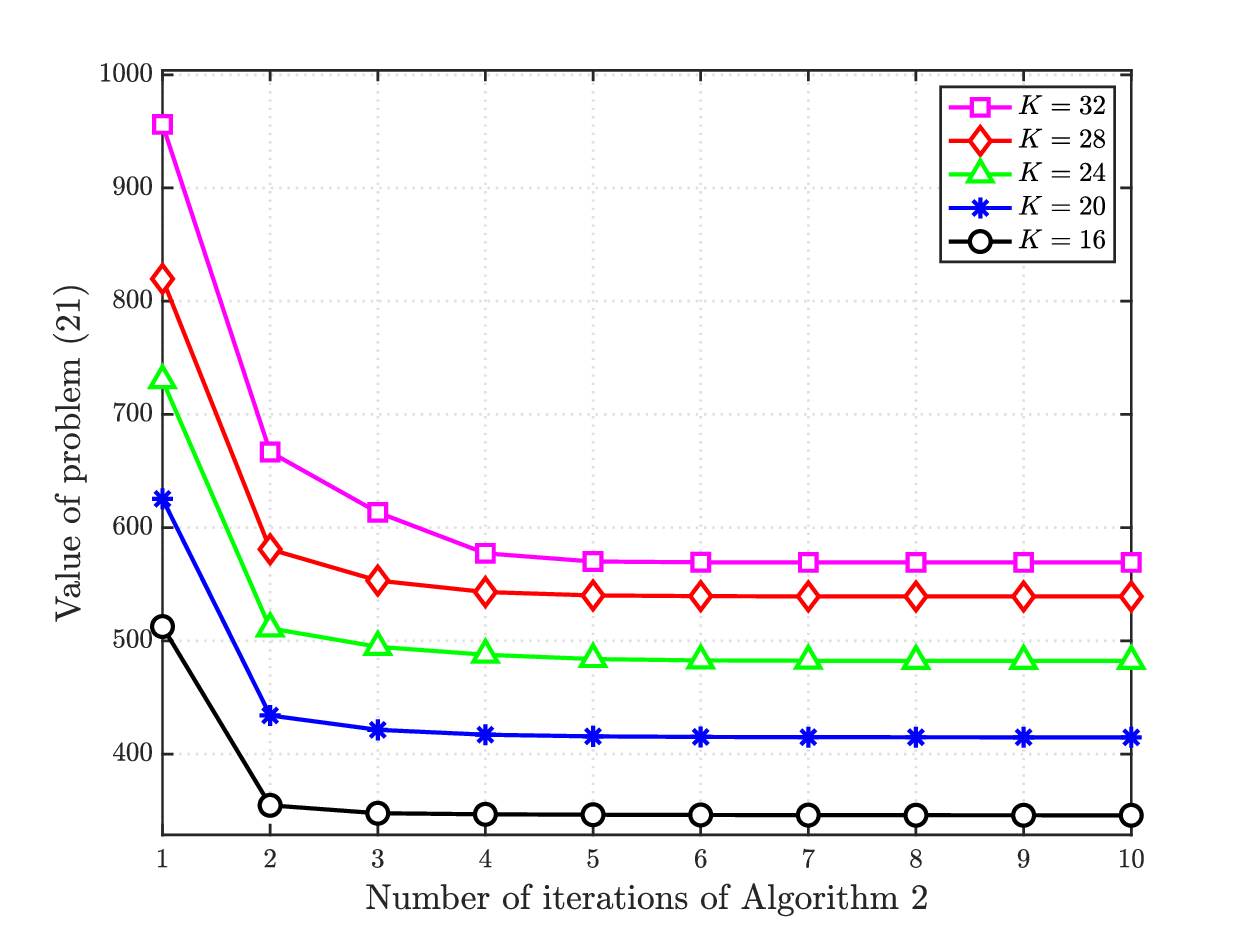}
	\caption{\textcolor{black}{The objective function value of problem (\ref{sub b}) versus the number of iterations of Algorithm~\ref{alg:cd}.}}
	\label{fig:func_conv}
\end{figure}

\subsection{Performance Comparison} 
Next, we demonstrate the superior performance of the proposed distributed algorithm against some benchmarks. For comparison, we consider four benchmarks, i.e., Mismatched CD~\cite{ganesan2021clustering}, Vanilla MLE~\cite{liu2024mle}, AMP~\cite{djelouat2021user}, and OAMP~\cite{cheng2020orthogonal}.
While Vanilla MLE, AMP, and OAMP are designed for single-cell activity detection, we extend them to cell-free massive MIMO. We describe these approaches as follows.
\begin{itemize}
\item Mismatched CD and Vanilla MLE are covariance-based approaches. Mismatched CD is the traditional rank-one update based approach~\cite{ganesan2021clustering}, which becomes inapplicable to hybrid near-far field scenarios because the joint distribution of near-field channels destroys the rank-one update structure. To implement this baseline, we directly utilize the traditional rank-one update by assuming all channels are far-field channels. On the other hand, Vanilla MLE still considers the near-field LoS component but ignores the  correlation structure of the channels.
     \item AMP and OAMP are CS-based approaches, which are designed for joint activity detection and channel estimation. OAMP is an improved version of AMP, which is applicable to more types of signature sequences. Notice that both approaches require the knowledge of the active probability of each device, while the covariance-based approaches do not require such prior information.
\end{itemize}

\begin{figure}[t]
	\centering
\includegraphics[width=0.85\columnwidth]{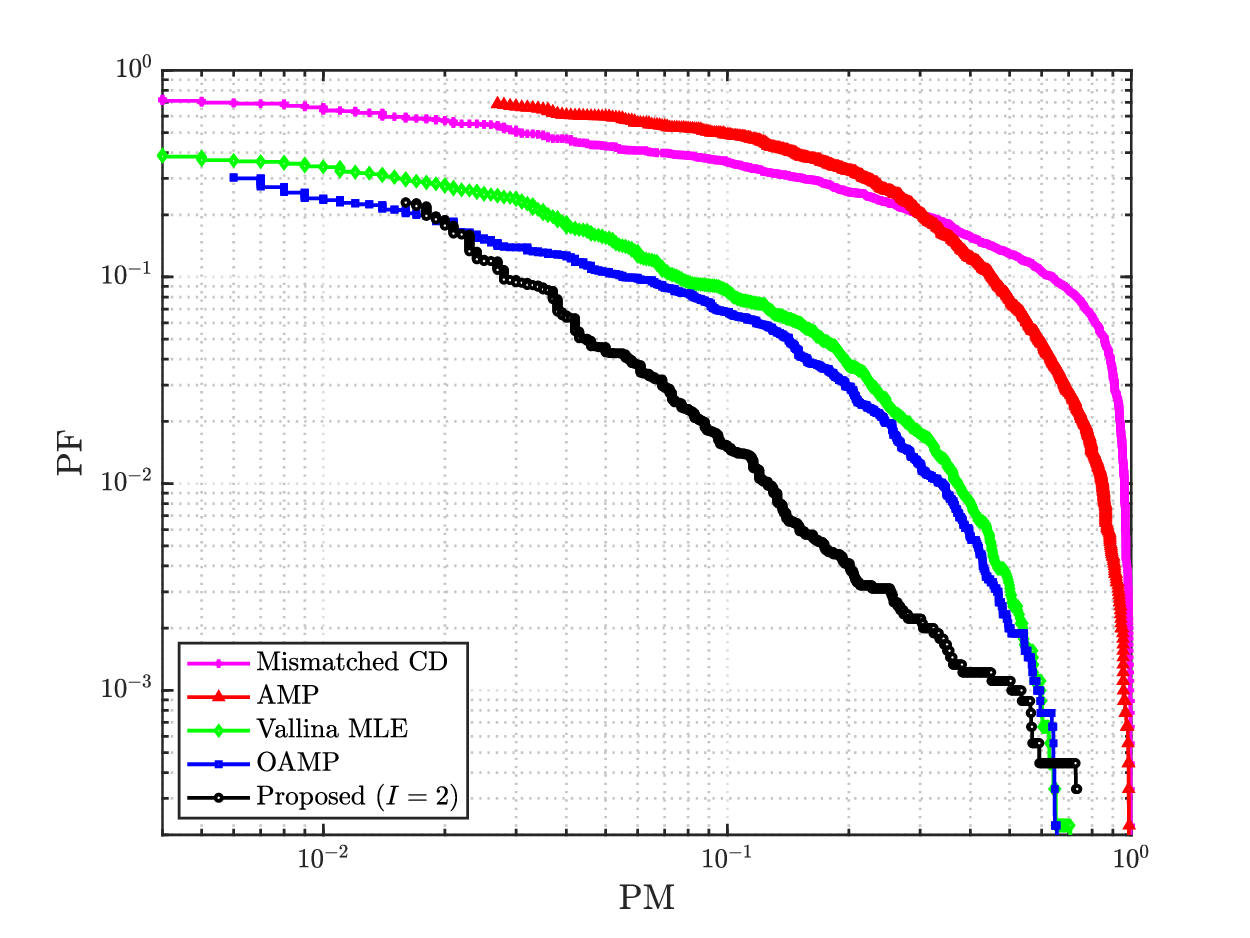}
	\caption{PM-PF trade-offs achieved by different  algorithms.}
	\label{PM-PF trade-offs}
	\vspace{-10pt}
\end{figure}

In Fig.~\ref{PM-PF trade-offs}, we compare the performance of different approaches in terms of PM and PF,
where $M=3$, $K=24$, and $L=6$. We can see that the proposed distributed algorithm outperforms all benchmarks. This superior performance stems from the accurate hybrid near-far field channel modeling and the covariance-based formulation that effectively exploits the spatial correlation structure of the hybrid near-far field channels. 
\begin{figure}[t]
	\centering
\includegraphics[width=0.85\columnwidth]{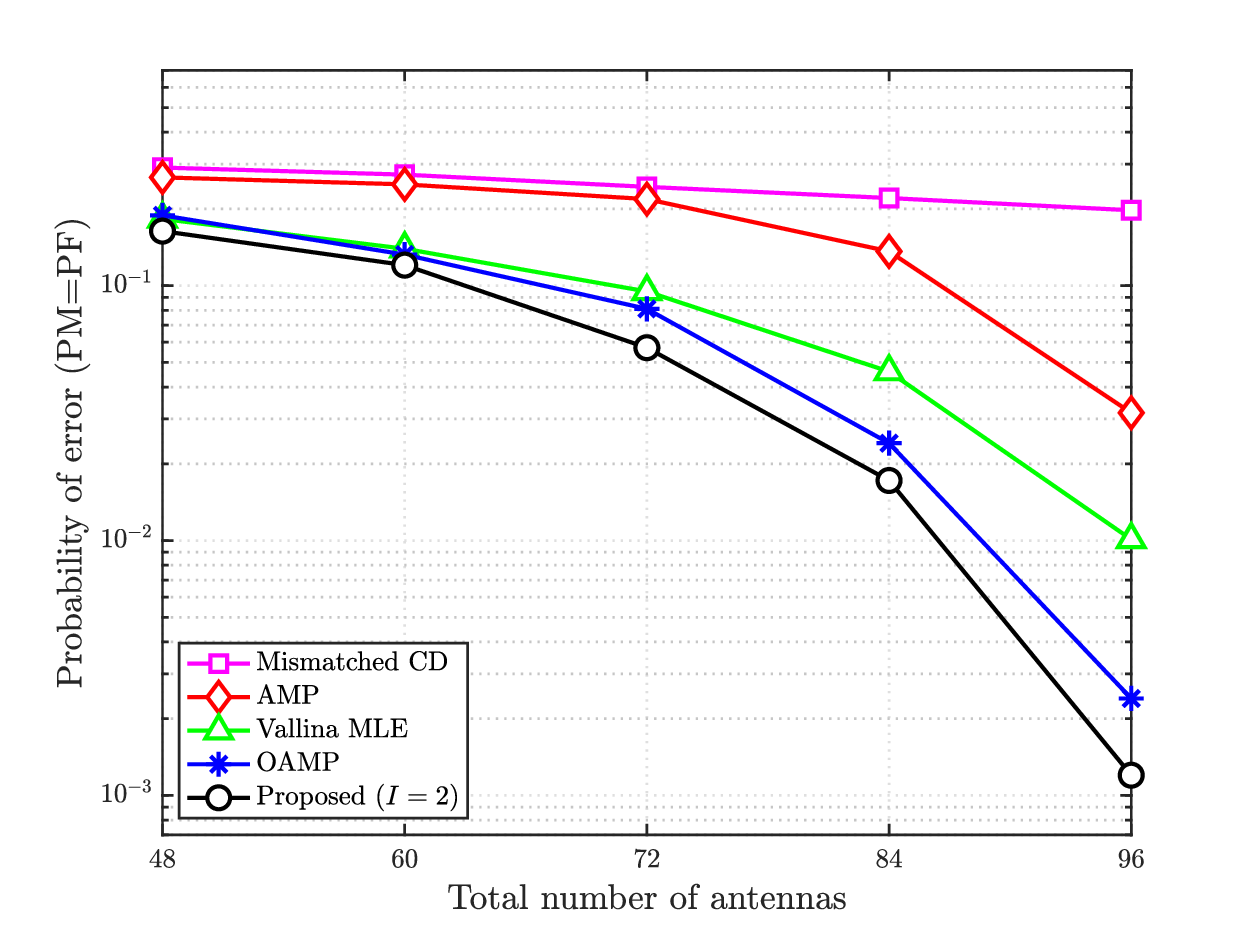}
	\caption{Probability of error versus total number of antennas.}
	\label{prob of error versus total number of antennas}
	\vspace{-10pt}
\end{figure}

Figure~\ref{prob of error versus total number of antennas} depicts the probability of error versus the total number of antennas,
where $M=3$ and $L=6$.
The total number of antennas varies from $48$ to $96$, with the antennas distributed equally among the APs. We can see that all approaches demonstrate improved detection performance as the total number of antennas increases,
with the proposed distributed algorithm consistently achieving superior performance. 

\begin{figure}[t]
	\centering
	\includegraphics[width=0.85\columnwidth]{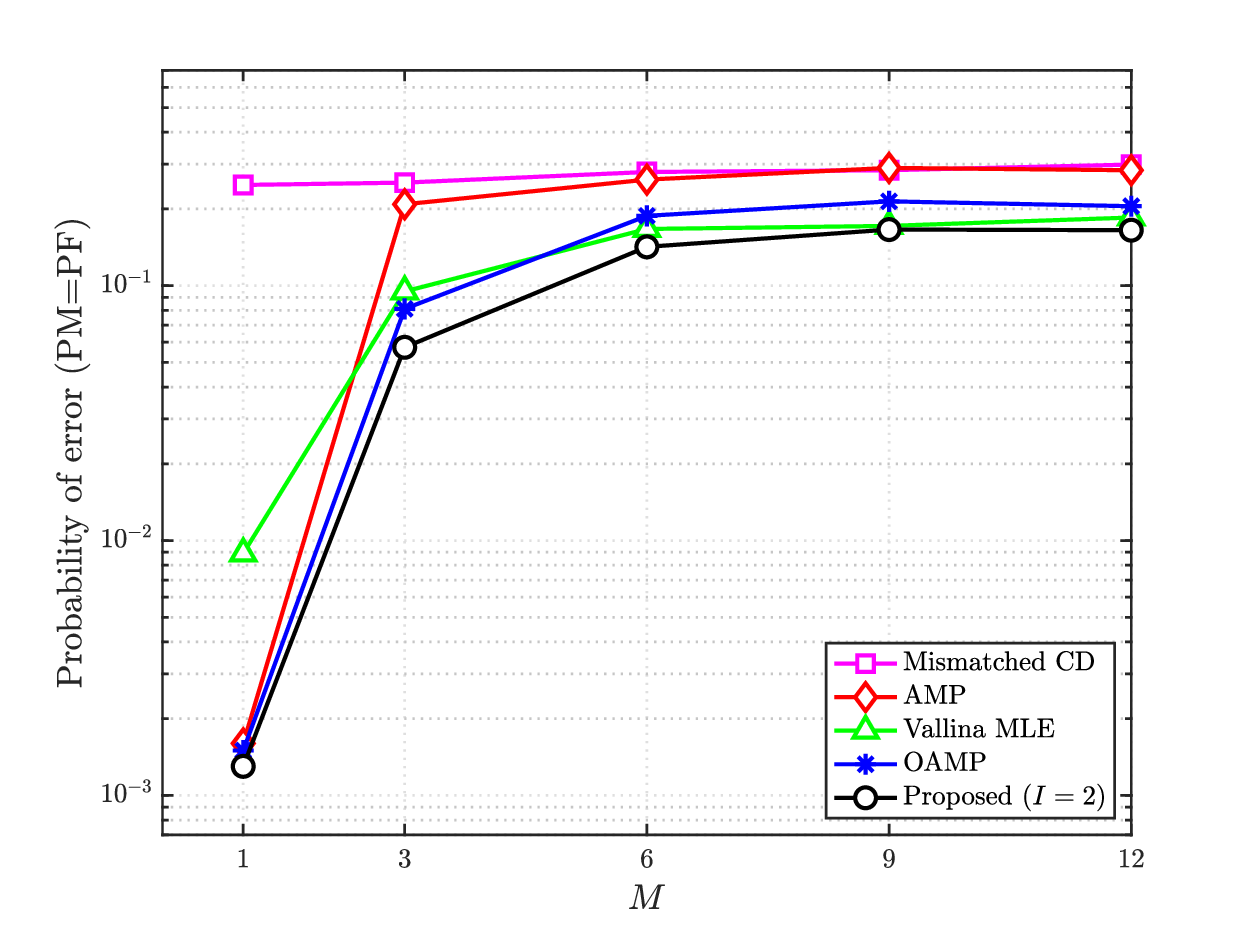}
	\caption{Probability of error versus $M$.}
	\label{prob of error versus number of APs}
	\vspace{-10pt}
\end{figure}

In addition to varying the total number of antennas, we further investigate the impact of the number of APs on detection performance.
Figure~\ref{prob of error versus number of APs} illustrates the probability of error versus the number of APs,
where the total number of antennas is fixed as $72$, and the signature sequence length is $L=6$. The number of APs varies from $1$ to $12$, with the antennas distributed equally among the APs. 
We can see that the proposed distributed algorithm consistently outperforms all benchmarks across different numbers of APs. Notably, all algorithms (except Mismatched CD) achieve their best performance when $M=1$, with performance degrading as $M$ increases.
This trend results from the trade-off between two competing factors: while increasing $M$ reduces the device-to-AP distance, which results in higher  SNR at each AP, it simultaneously decreases the Rayleigh distance at each AP due to fewer antennas per AP.
The reduction in near-field coverage outweighs the SNR benefits, leading to the overall performance degradation. We will further discuss the advantage of enlarging the near-field coverage in Section~\ref{sec:advantages_of_hybrid_near_far_field_channels}.

\begin{figure}[t]
	\centering
	\includegraphics[width=0.85\columnwidth]{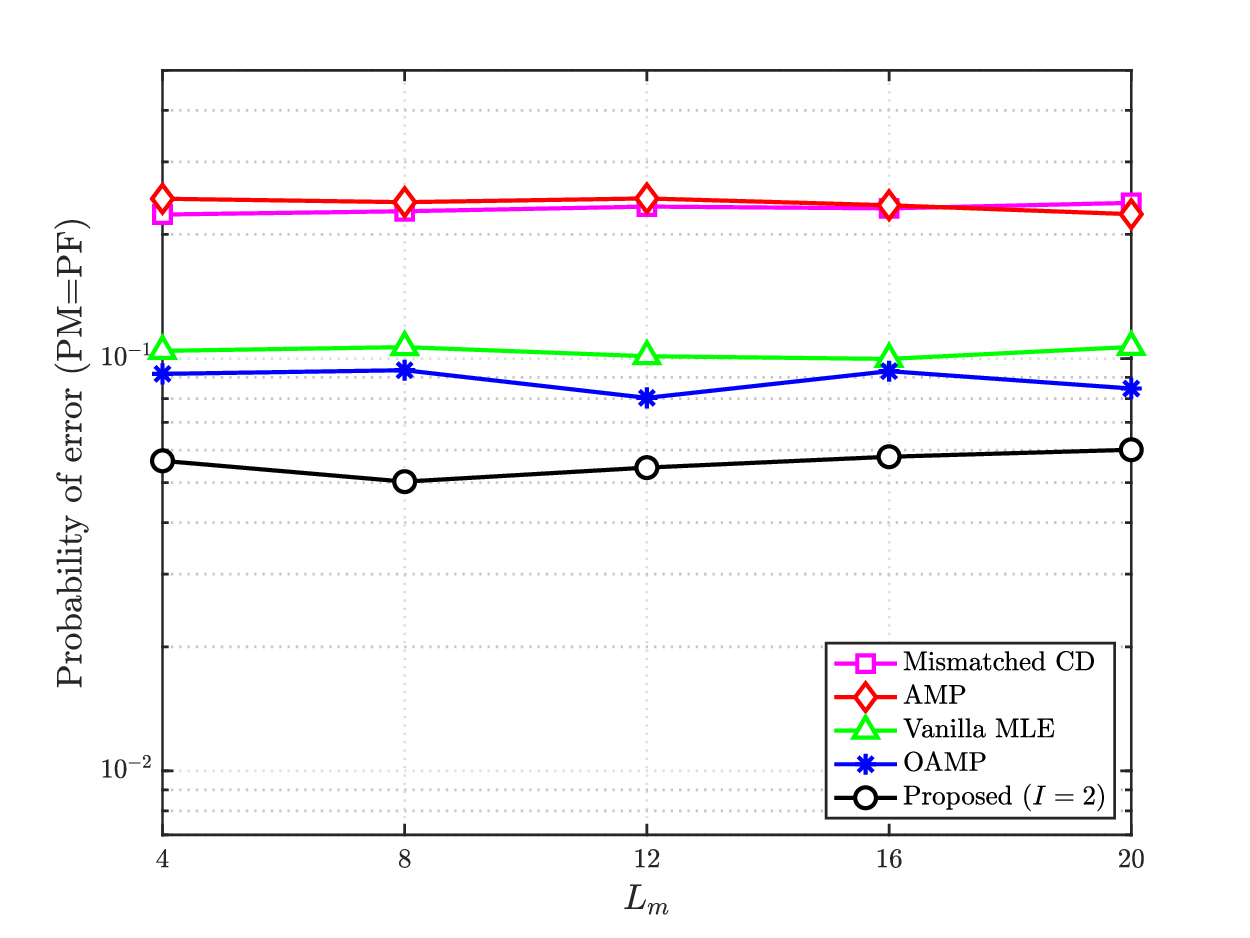}
	\caption{\textcolor{black}{Probability of error versus the number of scatterers $L_m$.}}
	\label{fig:scatter}
	\vspace{-10pt}
\end{figure}

\textcolor{black}{Furthermore, Fig.~\ref{fig:scatter} illustrates the probability of error versus the number of scatterers $L_m$, where $M=3$, $K=24$, and $L=6$. It can be observed that the proposed distributed algorithm consistently outperforms all benchmarks across different values of $L_m$, demonstrating the robustness of the proposed algorithm with respect to the number of scatterers.}

\begin{figure}[t]
	\centering
	\includegraphics[width=0.85\columnwidth]{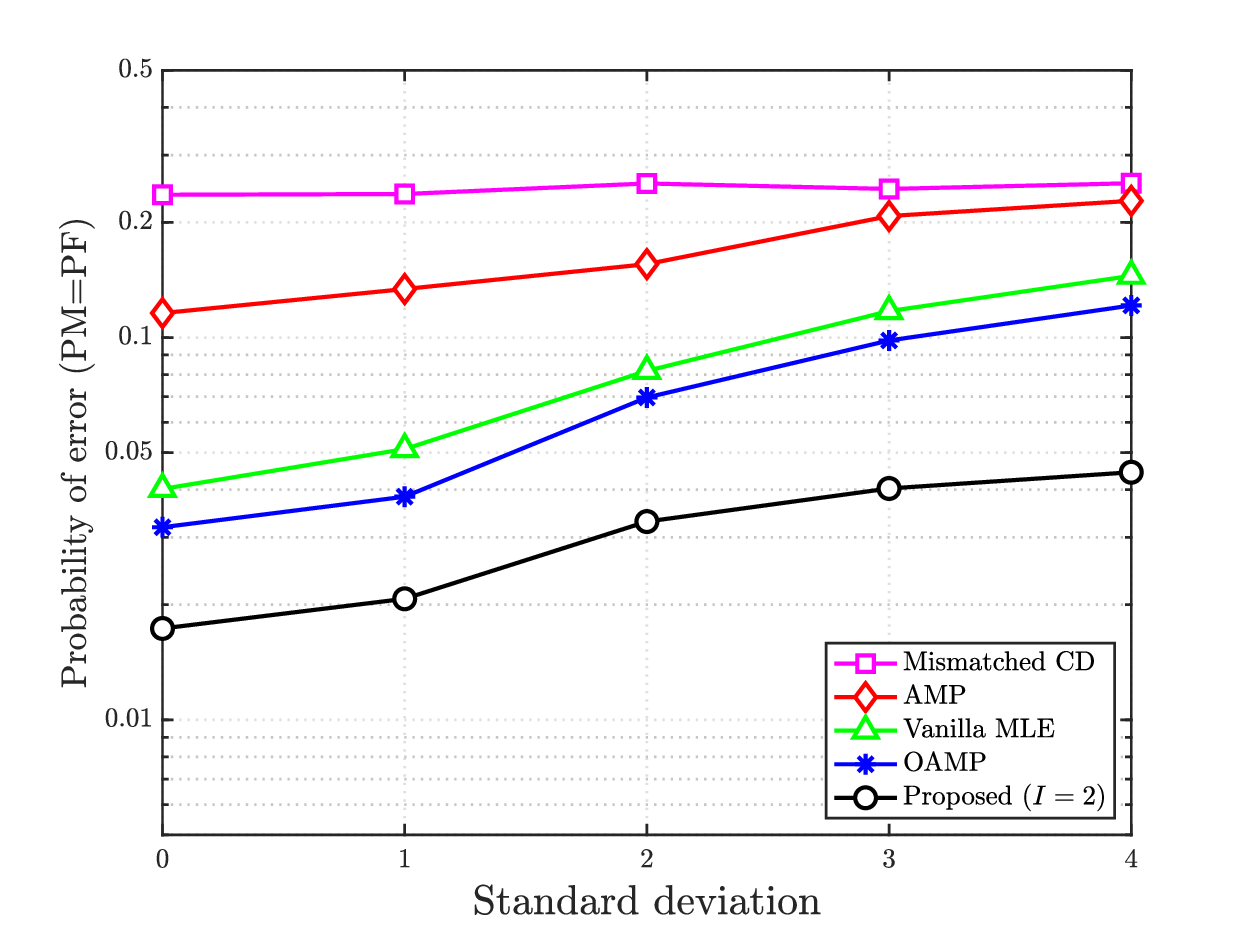}
	\caption{\textcolor{black}{Probability of error versus standard deviation $\sigma_{\text{loc}}$ of the location error.}}
	\label{fig:loc_error}
\end{figure}

\textcolor{black}{In practice, the device locations obtained during the registration phase may contain small errors, which are typically within a few meters with modern localization techniques \cite{peral2018survey}. To comprehensively evaluate the robustness of the proposed method, we provide additional simulation results in Fig.~\ref{fig:loc_error}. Specifically, we model the imperfect location by adding a Gaussian perturbation to the true position with standard deviation $\sigma_{\text{loc}}$ in meters. As shown in Fig.~\ref{fig:loc_error}, where $M=3$, $K=28$, and $L=6$, the proposed algorithm consistently outperforms all benchmarks across different levels of location error, demonstrating the robustness of the proposed algorithm with respect to location uncertainty.}

\subsection{ Computational Complexity and Communication Overhead Comparison}
\label{sec:complexity_comparison}

\textcolor{black}{Table~\ref{tab:complexity_benchmarks} summarizes the computational complexity and the runtime of all algorithms when $K=24$, $L=6$, and $M=3$. All experiments are implemented in MATLAB R2020a  on a computer equipped with an Intel Core i7-13700 CPU and 16 GB of memory. Specifically, the proposed algorithm requires $\mathcal{O}\!\left((LK)^2 \sum_{n=1}^{N} J_{m,n}\right)$ for all coordinate updates, which is independent of $M$ due to the parallel computation across different APs. Mismatched CD treats all channels as far-field and directly applies a rank-one update to the $L \times L$ covariance matrix, yielding $\mathcal{O}(MNL^2)$~\cite{li2022asynchronous}, \cite{ganesan2021clustering}. Vanilla MLE shares a similar update procedure as the proposed algorithm but operates in a centralized manner across $M$ APs with full-rank updates (i.e., $J_{m,n}=K$), which upon substituting into the Sherman--Morrison--Woodbury formula yields $\mathcal{O}(MNL^2K^3)$~\cite{Wang2026}, \cite{liu2024mle}. AMP involves standard matrix-vector multiplications with complexity $\mathcal{O}(MNLK)$~\cite{djelouat2021user}. OAMP involves a linear detection step that exploits the fast Hadamard transform to avoid explicit matrix multiplications, yielding $\mathcal{O}(MNK\log_2 N)$~\cite{cheng2020orthogonal}, and a nonlinear detection step dominated by a $K \times K$ matrix inversion per AP-device pair, yielding $\mathcal{O}(MNK^3)$~\cite{cheng2020orthogonal}. We can see that although the near-field channels impose much heavier computational burdens on the activity detection problem, the proposed algorithm does not incur much more complexity because of the parallel computation across multiple APs.}

\begin{table}[t]
\centering
\caption{\textcolor{black}{Computational Complexity Comparison}}
\label{tab:complexity_benchmarks}
\begin{tabular}{|c|c|c|}
\hline
\textbf{Algorithm} & \textbf{Complexity} & \textbf{Runtime (s)} \\
\hline
Proposed &  $\mathcal{O}\!\left((LK)^2 \sum_{n=1}^{N} J_{m,n}\right)$ & 1.62 \\
\hline
Mismatched CD & $\mathcal{O}\big(MNL^2\big)$ & 0.04\\
\hline
Vanilla MLE & $\mathcal{O}\big(MNL^2K^3\big)$ & 1.77 \\%
\hline
AMP & $\mathcal{O}\big(MNLK\big)$ & 0.69  \\
\hline
OAMP & $\mathcal{O}\big(MNK(\log_2 N+K^2)\big)$ & 1.52\\
\hline
\end{tabular}%
\end{table}

Table~\ref{tab:runtime_comparison} further demonstrates the parallel processing advantage of the proposed distributed algorithm by comparing the total computation time  across different AP numbers with $K=32$, and $L=8$. We can see that as $M$ increases, the total computation time of the centralized approach grows dramatically from 7.41 s to 114.96 s due to the increasing computational load at the CPU. In contrast, the distributed algorithm maintains relatively stable computation time from 2.59 s to 7.82 s by leveraging parallel computation. The speedup increases significantly with more APs, achieving a remarkable 14.70× improvement when $M=12$. This clearly demonstrates that the distributed algorithm's advantage becomes more pronounced as the system scales up.

\begin{table}[t]
	\centering
	\caption{Computation Time Comparison}
	\label{tab:runtime_comparison}
	\vspace{-5pt}
	\begin{tabular}{|c|c|c|c|}
	\hline
	$M$ & \textbf{Centralized (s)} & \textbf{Proposed (s)} & \textbf{Speedup} \\
	\hline
	3 & 7.41 & 2.59 & 2.86× \\
	\hline
	6  & 22.83 & 3.59 & 6.36× \\
	\hline
	9  & 52.59 & 5.35 & 9.83× \\
	\hline
	12 & 114.96 & 7.82 & 14.70× \\
	\hline
	\end{tabular}
	\vspace{-10pt}
	\end{table}
	\,

Figure~\ref{transmission bit} illustrates the probability of error versus the total number of bits transmitted over fronthaul links  with $M=3$, $K=24$, and $L=6$. For the centralized approach, we consider different numbers of bits from $ \{6, 8, 10\}$ for quantizing each entry of $\mathbf{Y}_m$, while for the proposed distributed algorithm, we can set a much smaller number of bits for quantization each entry, since the dynamic range $[0,1]$ for each entry of $\mathbf{a}$ and $\boldsymbol{\theta}_m$ is much smaller. In particular, the number of quantization bits per entry is set as $3$, and each AP transmits the local detection results of 60, 80, and 100 devices with the largest SNR.  We can see that the proposed distributed algorithm requires substantially lower communication overhead  while approximating the ideal performance assuming perfect transmission over the fronthaul links without quantization.

\begin{figure}[t]
	\centering
	\includegraphics[width=0.85\columnwidth]{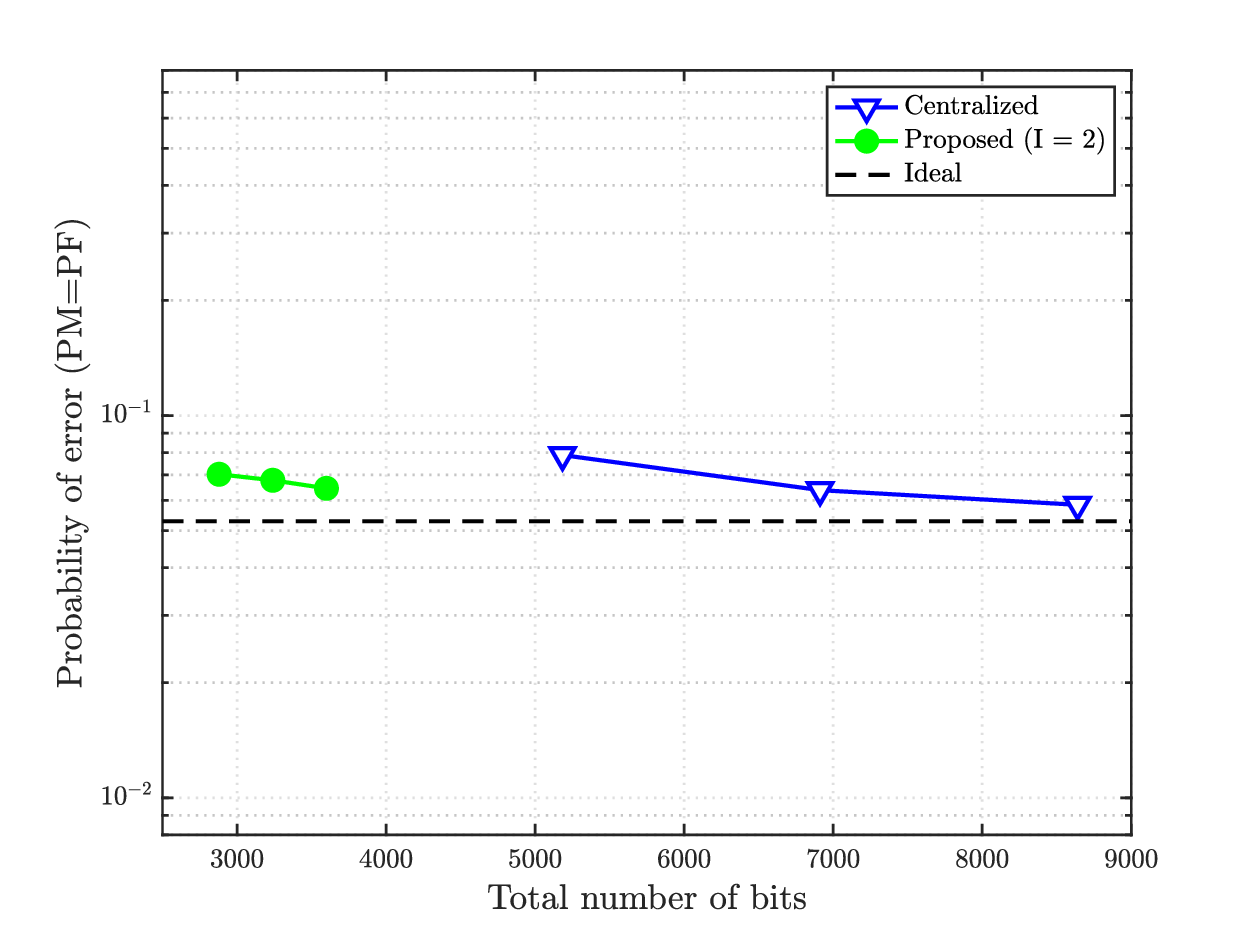}
	\caption{Probability of error versus total number of bits.}
	\label{transmission bit}
	\vspace{-10pt}
\end{figure}

\subsection{Advantages of Hybrid Near-Far Field Channels} 
\label{sec:advantages_of_hybrid_near_far_field_channels}
\begin{figure}[t]
	\centering
	\includegraphics[width=0.85\columnwidth]{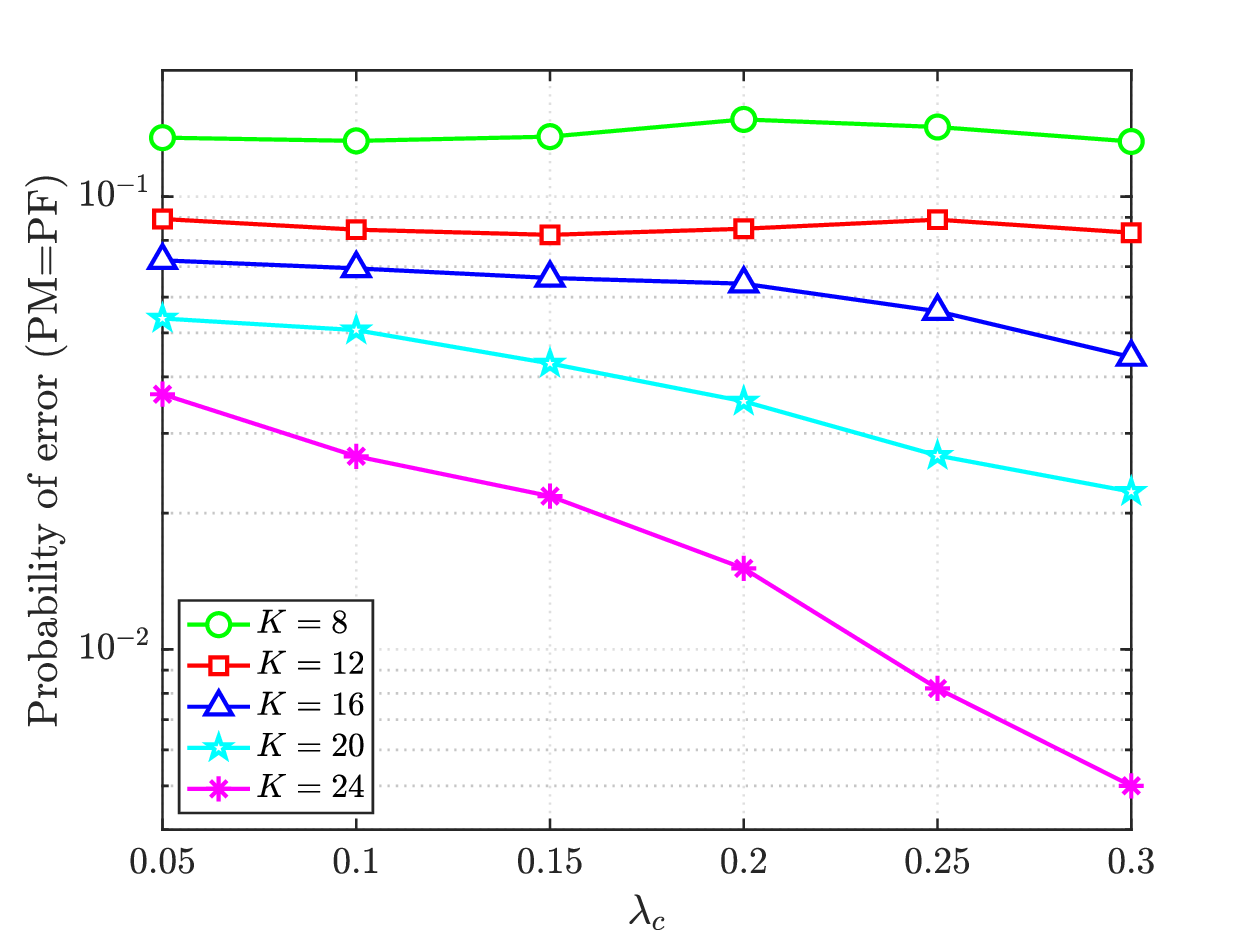}
	\caption{Probability of error versus $\lambda_c$.}
	\label{prob of error versus lambda}
	\vspace{-10pt}
\end{figure}

\begin{figure}[t]
	\centering
	\begin{subfigure}[t]{0.26\textwidth}
		\includegraphics[width=\textwidth]{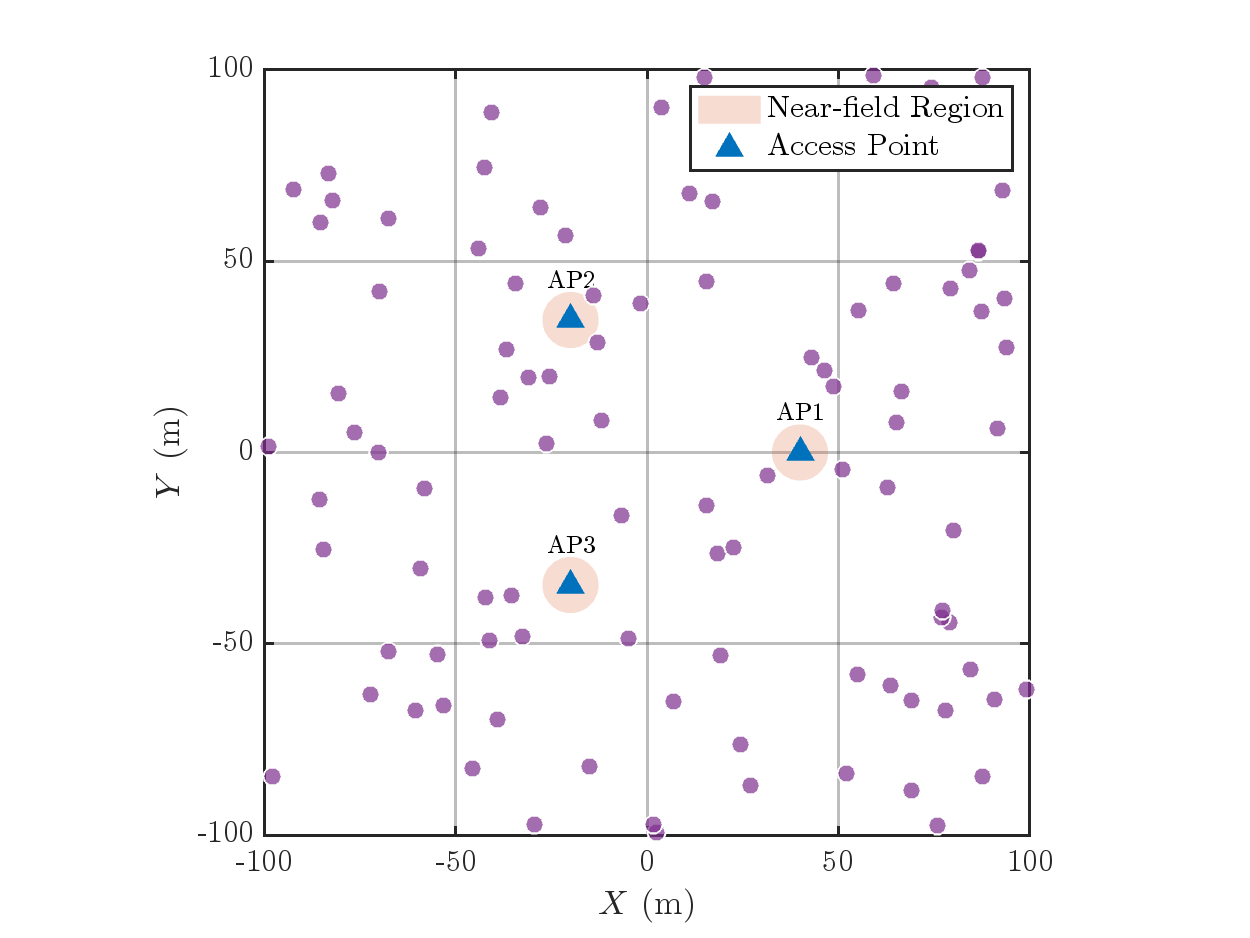}
		\caption{}
		\label{subfig:small_nfc}
	\end{subfigure}
	\hspace{-25pt}
	\begin{subfigure}[t]{0.26\textwidth}
		\includegraphics[width=\textwidth]{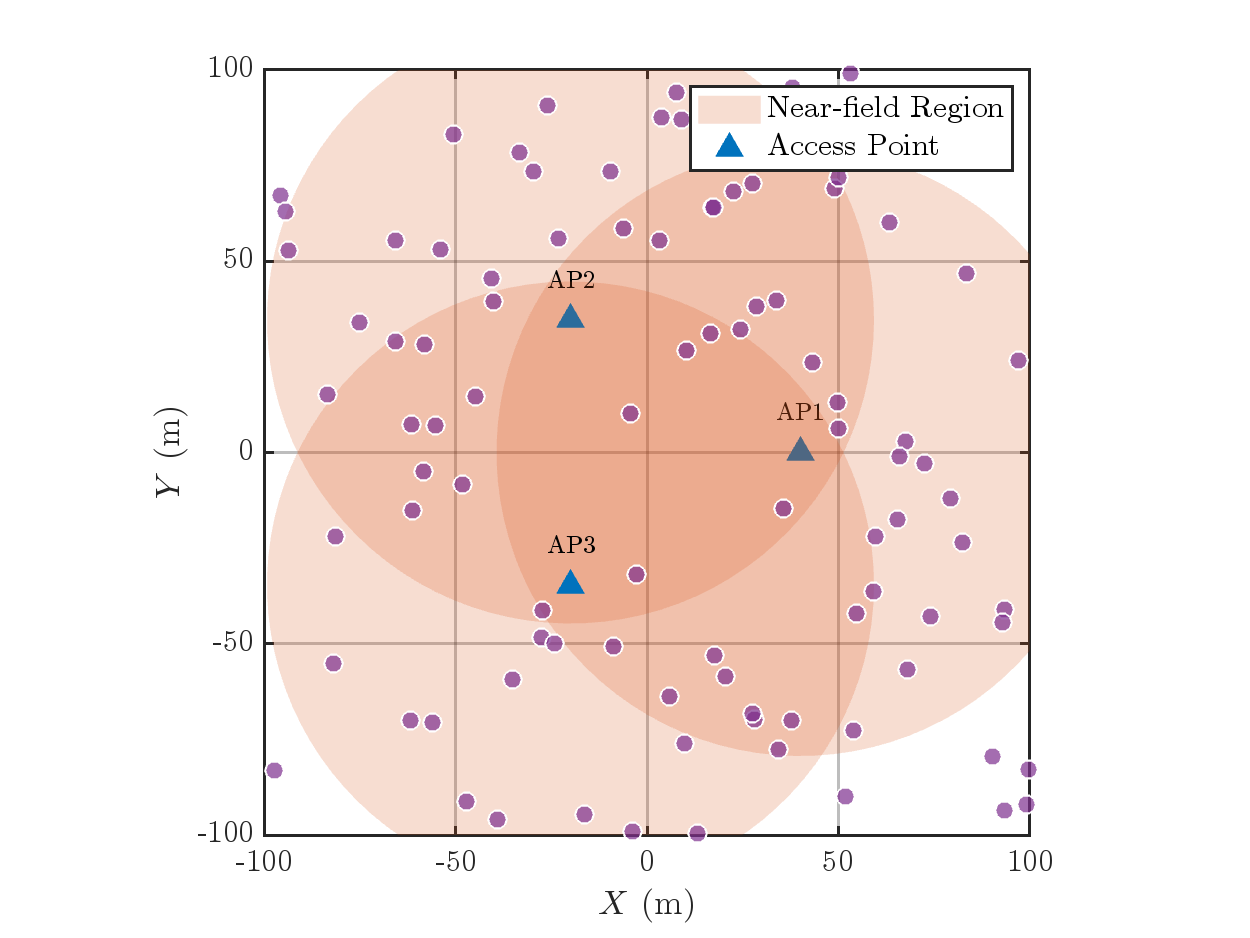}
		\caption{}
		\label{subfig:large_nfc}
	\end{subfigure}
	\caption{Visualization of near-field coverage in the $200 \times 200$ m$^2$ area. (a) Rayleigh distance $7$ m with $K=8$ and $\lambda_c=0.3$ m. (b) Rayleigh distance $79$ m with $K=24$ and $\lambda_c=0.3$ m.}
	\label{fig:nfc_coverage}
	\vspace{-10pt}
\end{figure}
Since the near-field channels are closely related to the Rayleigh distance $2D^{2}/\lambda_\text{c}$, where $D=(K-1)\lambda_\text{c}/2$, it is essential to explore the impact of both $\lambda_\text{c}$ and  $K$ on detection performance. To gain deeper insights into this relationship, Fig.~\ref{prob of error versus lambda} illustrates the probability of error versus $\lambda_\text{c}$ for different $K$'s, where  $M=3$ and $L=8$. When $K=8$ or $12$, it can be seen that the detection performance remains relatively stable across different wavelengths, showing minimal sensitivity to $\lambda_\text{c}$ variations.
However, for larger $K=16, 20, 24$, increasing $\lambda_\text{c}$ leads to significant performance improvements. This performance behavior is fundamentally explained by the expansion of near-field coverage, as visualized in Fig.~\ref{fig:nfc_coverage}.
The figure compares two scenarios: (a) a limited near-field coverage with Rayleigh distance of $7$ m when $K=8$ and $\lambda_\text{c}=0.3$ m, and (b) an extensive near-field coverage with Rayleigh distance of $79$ m when $K=24$ and $\lambda_\text{c}=0.3$ m.
The stark contrast between these scenarios demonstrates that when increasing the wavelength $\lambda_\text{c}$ can significantly expands the near-field region, it will incorporate more near-field channels, leading to improved detection performance, as suggested by Proposition~\ref{prop:mutual_similarity}.

\begin{figure}[t]
	\centering
	\includegraphics[width=0.85\columnwidth]{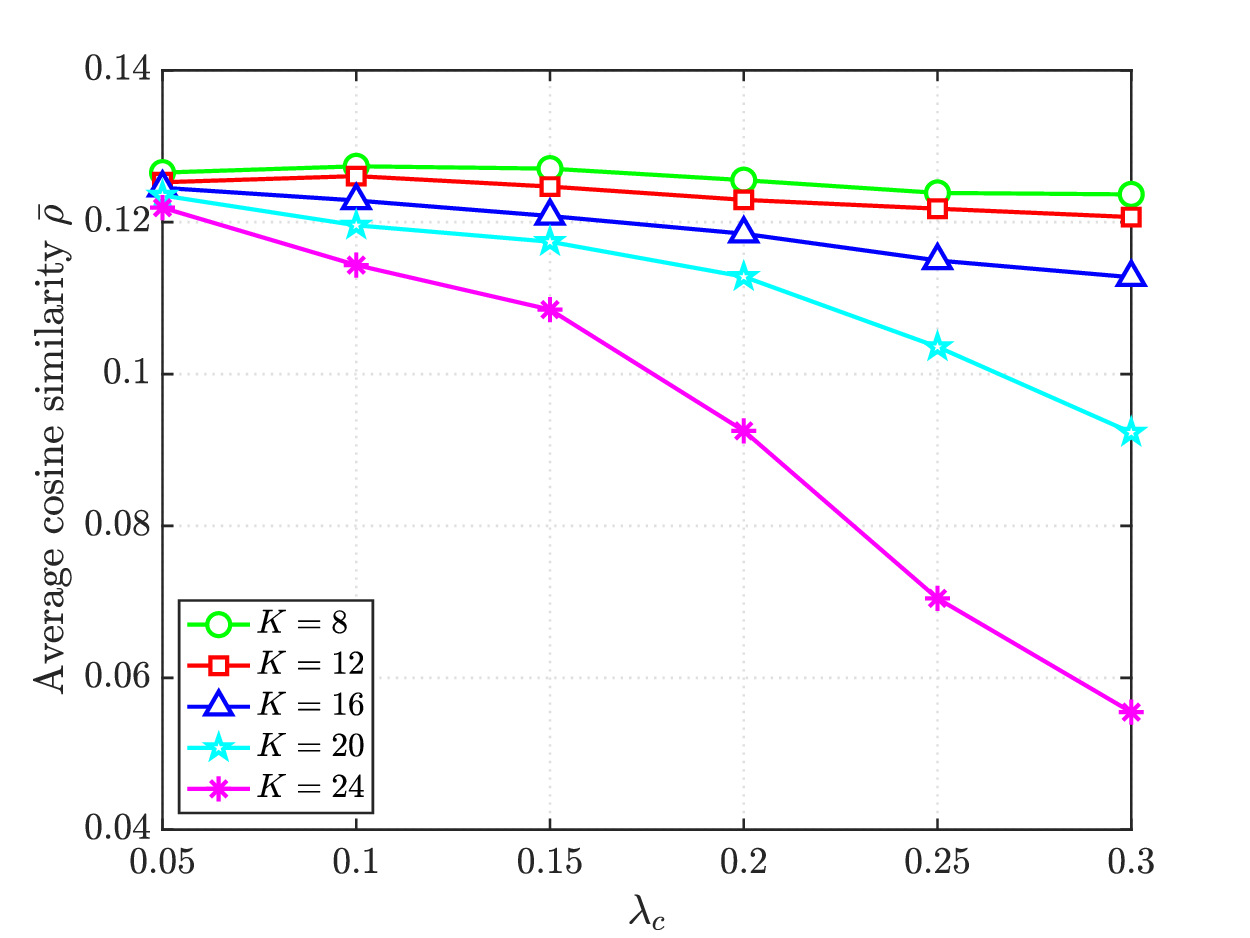}
	\caption{\textcolor{black}{Average cosine similarity $\bar{\rho}$ versus $\lambda_\text{c}$.}}
	\label{fig:cos_simi}
\end{figure}

\textcolor{black}{To validate Proposition~\ref{prop:mutual_similarity}, we define the average cosine similarity $\bar{\rho} = \frac{1}{N(N-1)} \sum_{n \neq n'} \rho_{n,n'}$, where $\rho_{n,n'}$ is the cosine similarity between the columns of $\boldsymbol{\Psi}_m$ defined in~(\ref{pho similarity}). Figure~\ref{fig:cos_simi}  plots $\bar{\rho}$ versus the carrier wavelength $\lambda_\text{c}$ under different numbers of antennas, where $M=3$ and $L=8$. In particular, as $\lambda_\text{c}$ increases from 0.05 to 0.3, the Rayleigh distance increases from 13 m to 79~m when the antenna number $K = 24$, while the Rayleigh distance increases from 1 m to 7~m when $K = 8$. It can be observed from this figure that as $\lambda_\text{c}$ increases, which means that the near-field proportion increases, $\bar{\rho}$ decreases under different numbers of antennas. Especially, when $K$ is larger, since the near-field percentage increases more dramatically as $\lambda_\text{c}$ increases, $\bar{\rho}$ decreases more notably. This is consistent with Proposition~\ref{prop:mutual_similarity}, which establishes that near-field channels yield smaller cosine similarities. The similar trends in Figures~\ref{prob of error versus lambda} and \ref{fig:cos_simi} both justify that the cosine similarity (\ref{pho similarity}) used in Proposition~\ref{prop:mutual_similarity} is a good indicator of detection performance.}

\begin{figure*}
	\centering
    	\begin{subfigure}[t]{0.325\textwidth}
		\includegraphics[width=\textwidth]{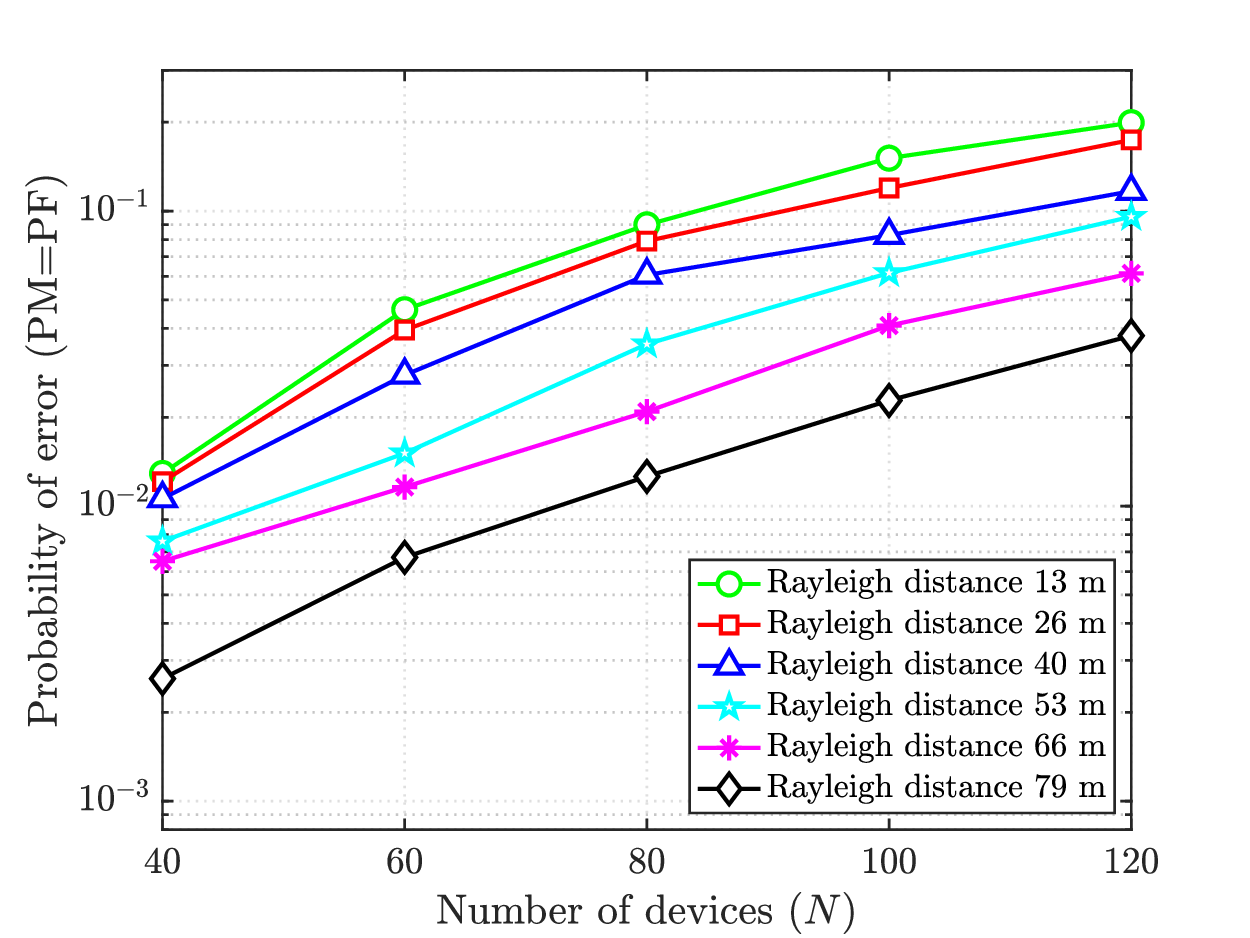}
		\caption{Probability of error versus $N$, when $L=6$ and $\alpha=0.1$.}
		\label{subfig-a}
	\end{subfigure}
	\begin{subfigure}[t]{0.325\textwidth}
		\includegraphics[width=\textwidth]{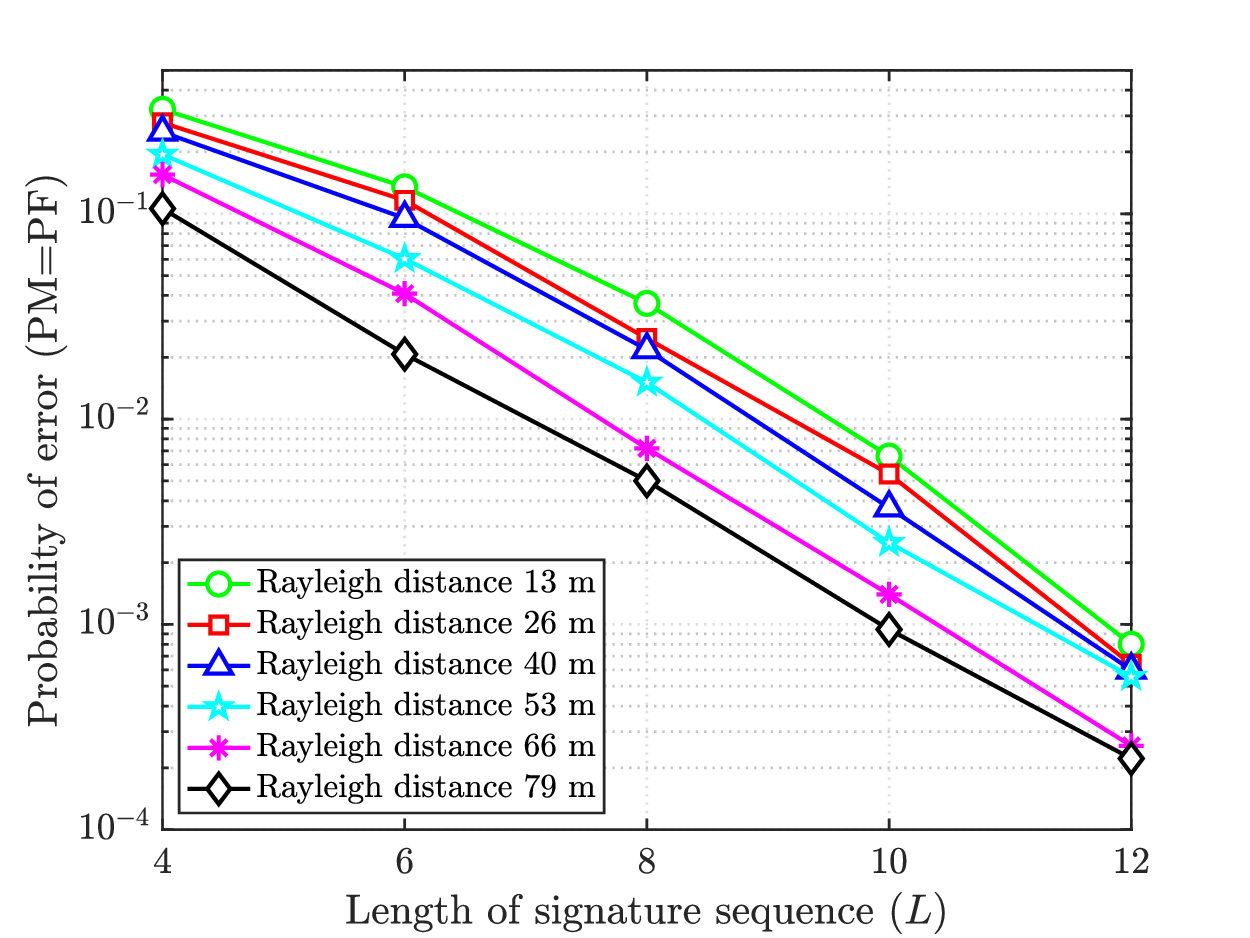}
		\caption{Probability of error versus $L$, when $N=100$ and $\alpha=0.1$}
		\label{subfig-b}
	\end{subfigure}
    \begin{subfigure}[t]{0.325\textwidth}
		\includegraphics[width=\textwidth]{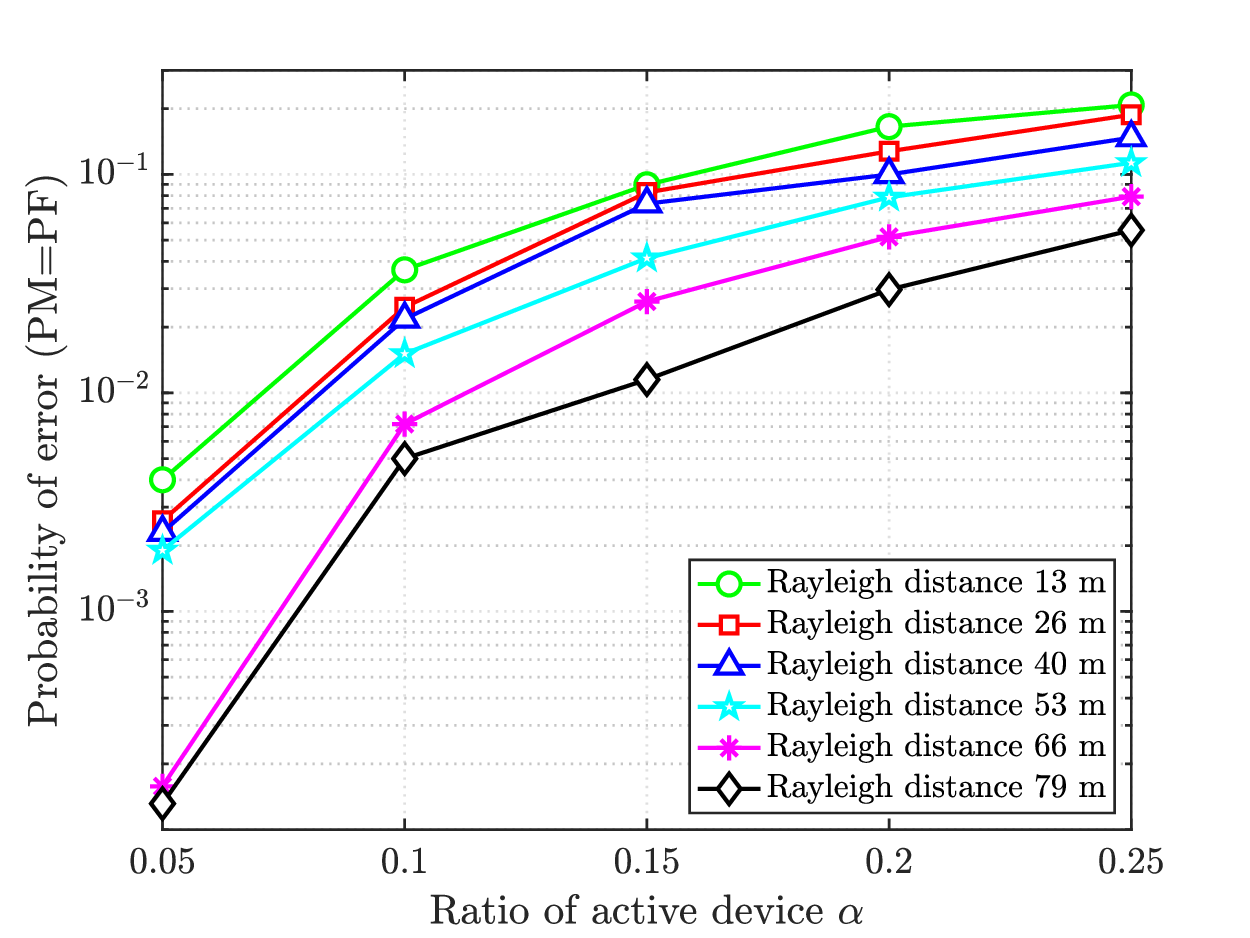}
		\caption{Probability of error versus $\alpha$, when $N=100$ and $L=8$.}
		\label{subfig-c}
	\end{subfigure}
	\caption{Performance comparison across different Rayleigh distances.}
	\label{analysis simulation}
	\vspace{-5pt}
\end{figure*}
Finally, Figure~\ref{analysis simulation} comprehensively illustrates the detection performance as a function of the number of devices, signature sequence length, and the active device ratio  across different Rayleigh distances, where  $M = 3$, and $K = 24$. The Rayleigh distance is $13$ m, $26$ m, $40$ m, $53$ m, $66$ m, and $79$ m, corresponding to $\lambda_\text{c}=0.05$~m, $0.1$ m, $0.15$ m, $0.20$, $0.25$, and $0.3$ m, respectively.  It can be observed that a
longer Rayleigh distance, which incorporates more near-field
channels, leads to better detection performance under different
system parameters, further validating the theoretical analysis in Section~\ref{sec:detection_performance}.

\section{Conclusion}
\label{sec:conclusion}
This paper investigated hybrid near-far field activity detection in cell-free massive MIMO systems. We first established a covariance-based formulation that captures the statistical property of the hybrid near-far field channels. Theoretical analysis demonstrated that increasing the proportion of near-field channels enhances the detection performance. Then, a unified distributed  algorithm was proposed with convergence guarantees. Extensive simulation results corroborated the theoretical analysis and showed that the proposed distributed algorithm achieves detection performance close to that of the centralized approach within just a few iterations while requiring much shorter computation time and lower communication overhead.

\appendices
\section{Proof of Proposition~\ref{prop:mutual_similarity}}
\label{app1}
For devices $n$ and $n'$, the cosine similarity is given by
\begin{align}
\label{prove proposition1}
    &\frac{\boldsymbol{\psi}_{m,n}^{\text{H}} \boldsymbol{\psi}_{m,n'}}{\|\boldsymbol{\psi}_{m,n}\|_2 \|\boldsymbol{\psi}_{m,n'}\|_2} \nonumber \\
	=&\frac{\text{tr}(\mathbf{\Xi}_{m,n} \mathbf{\Xi}_{m,n'})}{\|\mathbf{\Xi}_{m,n}\|_F \|\mathbf{\Xi}_{m,n'}\|_F} \times \left(\frac{|\mathbf{s}_n^{\text{H}} \mathbf{s}_{n'}|}{\|\mathbf{s}_n\|_2 \|\mathbf{s}_{n'}\|_2}\right)^2.
\end{align}
Let us define $\hat{\rho}_{n,n'} \triangleq \frac{\text{tr}(\mathbf{\Xi}_{m,n} \mathbf{\Xi}_{m,n'})}{\|\mathbf{\Xi}_{m,n}\|_F \|\mathbf{\Xi}_{m,n'}\|_F}$ as the matrix correlation factor. Based on (\ref{prove proposition1}), to compare the cosine similarity across different cases, it suffices to compare $\hat{\rho}_{n,n'}$ for different cases.

\textbf{Case 1: Two far-field devices:} 
With $\mathbf{\Xi}_{m,n} = g_{m,n}\mathbf{I}_K$ and $\mathbf{\Xi}_{m,n'} = g_{m,n'}\mathbf{I}_K$, we have
\begin{align}
\hat{\rho}_{\text{FF-FF}} = \frac{\text{tr}(g_{m,n}\mathbf{I}_K \cdot g_{m,n'}\mathbf{I}_K)}{\|g_{m,n}\mathbf{I}_K\|_F \|g_{m,n'}\mathbf{I}_K\|_F}  = 1.
\end{align}

\textbf{Case 2: Two near-field devices:} With $\mathbf{\Xi}_{m,n} = \mathbf{R}_{m,n}$ and $\mathbf{\Xi}_{m,n'} = \mathbf{R}_{m,n'}$, we have
\begin{align}
	\label{large eigenvalue}
\hat{\rho}_{\text{NF-NF}} = \frac{\text{tr}(\mathbf{R}_{m,n} \mathbf{R}_{m,n'})}{\|\mathbf{R}_{m,n}\|_F \|\mathbf{R}_{m,n'}\|_F}\leq 1;
\end{align}
where the inequality follows from the Cauchy-Schwarz inequality
for matrices.

\textbf{Case 3: One near-field and one far-field device:} With $\mathbf{\Xi}_{m,n} = \mathbf{R}_{m,n}$ and $\mathbf{\Xi}_{m,n'} = g_{m,n'}\mathbf{I}_K$, we have
\begin{align}
\hat{\rho}_{\text{NF-FF}} = \frac{\text{tr}(\mathbf{R}_{m,n} \cdot g_{m,n'}\mathbf{I}_K)}{\|\mathbf{R}_{m,n}\|_F \|g_{m,n'}\mathbf{I}_K\|_F}< 1,
\end{align}
where the inequality strictly holds due to $\mathbf{R}_{m,n}$ is not a diagonal matrix.

Combining the three cases, we have $\hat{\rho}_{\text{NF-NF}} \leq \hat{\rho}_{\text{FF-FF}}$, and $\hat{\rho}_{\text{NF-FF}} < \hat{\rho}_{\text{FF-FF}}$. Since the signature sequence factor $\left(\frac{|\mathbf{s}_n^{\text{H}} \mathbf{s}_{n'}|}{\|\mathbf{s}_n\|_2 \|\mathbf{s}_{n'}\|_2}\right)^2$ in (\ref{prove proposition1}) is identical across all cases, we have $\rho_{\text{NF-NF}} \leq \rho_{\text{FF-FF}}$, and $\rho_{\text{NF-FF}} < \rho_{\text{FF-FF}}$.

\section{\textcolor{black}{Detailed Expressions of $p_{\mathrm{approx}}(d)$ and $\hat{\rho}_{i,m}(\mathbf{a})$}}
\textcolor{black}{The detailed expressions of $p_{\mathrm{approx}}(d)$ and $\hat{\rho}_{i,m}(\mathbf{a})$ are given in (\ref{eq:distributed inexact}) and (\ref{eq:centralized_obj_terms}) on the top of the next page, respectively.}
\label{appB}
\begin{figure*}
	\vspace{-6pt}
\begin{align}\label{eq:distributed inexact}
&\quad\, \,\:\: p_{\mathrm{approx}}(d)=d \rho_{1}(\boldsymbol{\theta}_{m}) + d^2 \rho_{2}(\boldsymbol{\theta}_{m}) + d^3 \rho_{3}(\boldsymbol{\theta}_{m}) + d^4 \rho_{4}(\boldsymbol{\theta}_{m}), \nonumber\\
&\quad\quad\quad\rho_{1}(\boldsymbol{\theta}_{m}) = \text{tr} \left( \mathbf{X}_{m,n}^\text{H} \tilde{\mathbf{C}}_{m}^{-1} \mathbf{X}_{m,n} \right) - 2 \, \text{Re} \left( \left( \mathbf{y}_{m} - \tilde{\bar{\mathbf{y}}}_{m} \right)^\text{H} \tilde{\mathbf{C}}_{m}^{-1} \left( \boldsymbol{\varpi}_{m,n} \otimes \mathbf{s}_n \right) \right) \nonumber\\
&\quad\quad\quad\quad\quad\quad\quad - \left( \mathbf{y}_{m} - \tilde{\bar{\mathbf{y}}}_{m} \right)^\text{H} \tilde{\mathbf{C}}_{m}^{-1} \mathbf{X}_{m,n} \mathbf{X}_{m,n}^\text{H} \tilde{\mathbf{C}}_{m}^{-1} \left( \mathbf{y}_{m} - \tilde{\bar{\mathbf{y}}}_{m} \right) + \lambda_{m, n}^{(i-1)} + \mu \left(\theta_{m,n}-a_{n}^{(i-1)}\right), \nonumber\\
&\quad\quad\quad\rho_{2}(\boldsymbol{\theta}_{m}) = \left( \boldsymbol{\varpi}_{m,n} \otimes \mathbf{s}_n \right)^\text{H} \tilde{\mathbf{C}}_{m}^{-1} \left( \boldsymbol{\varpi}_{m,n} \otimes \mathbf{s}_n \right) + 2 \, \text{Re} \left( \left( \mathbf{y}_{m} - \tilde{\bar{\mathbf{y}}}_{m} \right)^\text{H} \tilde{\mathbf{C}}_{m}^{-1} \mathbf{X}_{m,n} \mathbf{X}_{m,n}^\text{H} \tilde{\mathbf{C}}_{m}^{-1} \left( \boldsymbol{\varpi}_{m,n} \otimes \mathbf{s}_n \right) \right) \nonumber\\
&\quad\quad\quad\quad\quad\quad\quad + \left( \mathbf{y}_{m} - \tilde{\bar{\mathbf{y}}}_{m} \right)^\text{H} \tilde{\mathbf{C}}_{m}^{-1} \mathbf{X}_{m,n} \mathbf{X}_{m,n}^\text{H} \tilde{\mathbf{C}}_{m}^{-1} \mathbf{X}_{m,n} \mathbf{X}_{m,n}^\text{H} \tilde{\mathbf{C}}_{m}^{-1} \left( \mathbf{y}_{m} - \tilde{\bar{\mathbf{y}}}_{m} \right)+ \frac{\mu}{2}, \nonumber\\
&\quad\quad\quad\rho_{3}(\boldsymbol{\theta}_{m}) = -2 \, \text{Re} \left( \left( \mathbf{y}_{m} - \tilde{\bar{\mathbf{y}}}_{m} \right)^\text{H} \tilde{\mathbf{C}}_{m}^{-1} \mathbf{X}_{m,n} \mathbf{X}_{m,n}^\text{H} \tilde{\mathbf{C}}_{m}^{-1} \mathbf{X}_{m,n} \mathbf{X}_{m,n}^\text{H} \tilde{\mathbf{C}}_{m}^{-1} \left( \boldsymbol{\varpi}_{m,n} \otimes \mathbf{s}_n \right) \right) \nonumber\\
&\quad\quad\quad\quad\quad\quad\quad - \left( \boldsymbol{\varpi}_{m,n} \otimes \mathbf{s}_n \right)^\text{H} \tilde{\mathbf{C}}_{m}^{-1} \mathbf{X}_{m,n} \mathbf{X}_{m,n}^\text{H} \tilde{\mathbf{C}}_{m}^{-1} \left( \boldsymbol{\varpi}_{m,n} \otimes \mathbf{s}_n \right), \nonumber\\
&\quad\quad\quad\rho_{4}(\boldsymbol{\theta}_{m}) = \left( \boldsymbol{\varpi}_{m,n} \otimes \mathbf{s}_n \right)^\text{H} \tilde{\mathbf{C}}_{m}^{-1} \mathbf{X}_{m,n} \mathbf{X}_{m,n}^\text{H} \tilde{\mathbf{C}}_{m}^{-1} \mathbf{X}_{m,n} \mathbf{X}_{m,n}^\text{H} \tilde{\mathbf{C}}_{m}^{-1} \left( \boldsymbol{\varpi}_{m,n} \otimes \mathbf{s}_n \right).
\end{align}
\vspace{-5pt}
\hrulefill
\end{figure*}

\begin{figure*}
	\vspace{-6pt}
\begin{align}
\label{eq:centralized_obj_terms}
\hat{\rho}_{1,m}(\mathbf{a}) &= \text{tr} \left( \mathbf{X}_{m,n}^\text{H} \mathbf{C}_{m}^{-1} \mathbf{X}_{m,n} \right) - 2 \, \text{Re} \left( \left( \mathbf{y}_{m} - \bar{\mathbf{y}}_{m} \right)^\text{H} \mathbf{C}_{m}^{-1} \left( 
	\boldsymbol{\varpi}_{m,n} \otimes \mathbf{s}_n \right) \right) \nonumber \\
&\quad - \left( \mathbf{y}_{m} - \bar{\mathbf{y}}_{m} \right)^\text{H} \mathbf{C}_{m}^{-1} \mathbf{X}_{m,n} \mathbf{X}_{m,n}^\text{H} \mathbf{C}_{m}^{-1} \left( \mathbf{y}_{m} - \bar{\mathbf{y}}_{m} \right), \nonumber \\
\hat{\rho}_{2,m}(\mathbf{a}) &= \left( \boldsymbol{\varpi}_{m,n} \otimes \mathbf{s}_n \right)^\text{H} \mathbf{C}_{m}^{-1} \left( \boldsymbol{\varpi}_{m,n} \otimes \mathbf{s}_n \right) + 2 \, \text{Re} \left( \left( \mathbf{y}_{m} - \bar{\mathbf{y}}_{m} \right)^\text{H} \mathbf{C}_{m}^{-1} \mathbf{X}_{m,n} \mathbf{X}_{m,n}^\text{H} \mathbf{C}_{m}^{-1} \left( \boldsymbol{\varpi}_{m,n} \otimes \mathbf{s}_n \right) \right) \nonumber \\
&\quad + \left( \mathbf{y}_{m} - \bar{\mathbf{y}}_{m} \right)^\text{H} \mathbf{C}_{m}^{-1} \mathbf{X}_{m,n} \mathbf{X}_{m,n}^\text{H} \mathbf{C}_{m}^{-1} \mathbf{X}_{m,n} \mathbf{X}_{m,n}^\text{H} \mathbf{C}_{m}^{-1} \left( \mathbf{y}_{m} - \bar{\mathbf{y}}_{m} \right), \nonumber \\
\hat{\rho}_{3,m}(\mathbf{a}) &= -2 \, \text{Re} \left( \left( \mathbf{y}_{m} - \bar{\mathbf{y}}_{m} \right)^\text{H} \mathbf{C}_{m}^{-1} \mathbf{X}_{m,n} \mathbf{X}_{m,n}^\text{H} \mathbf{C}_{m}^{-1} \mathbf{X}_{m,n} \mathbf{X}_{m,n}^\text{H} \mathbf{C}_{m}^{-1} \left( \boldsymbol{\varpi}_{m,n} \otimes \mathbf{s}_n \right) \right) \nonumber \\
&\quad - \left( \boldsymbol{\varpi}_{m,n} \otimes \mathbf{s}_n \right)^\text{H} \mathbf{C}_{m}^{-1} \mathbf{X}_{m,n} \mathbf{X}_{m,n}^\text{H} \mathbf{C}_{m}^{-1} \left( \boldsymbol{\varpi}_{m,n} \otimes \mathbf{s}_n \right), \nonumber \\
\hat{\rho}_{4,m}(\mathbf{a}) &= \left( \boldsymbol{\varpi}_{m,n} \otimes \mathbf{s}_n \right)^\text{H} \mathbf{C}_{m}^{-1} \mathbf{X}_{m,n} \mathbf{X}_{m,n}^\text{H} \mathbf{C}_{m}^{-1} \mathbf{X}_{m,n} \mathbf{X}_{m,n}^\text{H} \mathbf{C}_{m}^{-1} \left( \boldsymbol{\varpi}_{m,n} \otimes \mathbf{s}_n \right).
\end{align}
\vspace{-5pt}
\hrulefill
\end{figure*}

\section{Proof of Proposition~\ref{proposition1}}
\label{app2}
The outline of the proof is as follows. First, we prove that there exists a constant $\rho=\bar{\rho}_2+\bar{\rho}_3+\bar{\rho}_4$ such that $\bar{\rho}_2$, $\bar{\rho}_3$, and $\bar{\rho}_4$ are the upper bounds of $\rho_{2}(\boldsymbol{\theta}_{m})$, $\rho_{3}(\boldsymbol{\theta}_{m})$, and $\rho_{4}(\boldsymbol{\theta}_{m})$ in (\ref{eq:distributed inexact}), respectively. Then, we show that Algorithm~\ref{alg:cd} achieves a sufficient decrease of the objective function $U_m(\boldsymbol{\theta}_m)$. Finally, we demonstrate that Algorithm~\ref{alg:cd} guarantees the convergence to a stationary point of problem (\ref{sub b}).

Substituting  (\ref{eq:term1}) and (\ref{eq:term2}) into (\ref{sub b}), we have
\begin{equation}
	\label{eq:original}
	U_m(\boldsymbol{\theta}_m + d\mathbf{e}_n) = U_m(\boldsymbol{\theta}_m) + p_{\mathrm{approx}}(d) + o(d).
\end{equation}
The gradient of $U_m(\boldsymbol{\theta}_m)$ with respect to $\theta_{m,n}$ is given by
\begin{equation}
	\label{eq:gradient}
	\left. [\nabla U_m(\boldsymbol{\theta}_m)]_n = \frac{\partial U_m(\boldsymbol{\theta}_m + d\mathbf{e}_n)}{\partial \, d} \right|_{d = 0} = \rho_{1}(\boldsymbol{\theta}_{m}).
\end{equation}
Since $\rho_{1}(\boldsymbol{\theta}_{m})$ is continuous, and $\boldsymbol{\theta}_m$ is in the bounded region $[0,1]^N,$ according to (\ref{eq:gradient}), we can conclude that $\nabla U_m(\boldsymbol{\theta}_m)$ is Lipschitz continuous. That is, there exists a constant $L_m>0$ such that for any $\boldsymbol{\hat{\theta}}_m, \boldsymbol{\tilde{\theta}}_m \in [0,1]^N,$ the following inequality holds:
\begin{equation}
	\left\| \nabla U_m(\boldsymbol{\hat{\theta}}_m) - \nabla U_m(\boldsymbol{\tilde{\theta}}_m) \right\|_2 \le L_m \left\| \boldsymbol{\hat{\theta}}_m - \boldsymbol{\tilde{\theta}}_m \right\|_2.
\end{equation}
Consequently, we have a quadratic upper bound on $U_m(\boldsymbol{\theta}_m)$~\cite[Proposition A.24]{bertsekas1997nonlinear}:
\begin{align}
\label{eq:quadratic upper bound}
U_m(\boldsymbol{\tilde{\theta}}_m) \le U_m(\boldsymbol{\hat{\theta}}_m) &+ \left(\nabla U_m(\boldsymbol{\hat{\theta}}_m)\right)^T \left( \boldsymbol{\tilde{\theta}}_m - \boldsymbol{\hat{\theta}}_m \right) \notag\\ 
	&+ \frac{L_m}{2} \left\| \boldsymbol{\tilde{\theta}}_m - \boldsymbol{\hat{\theta}}_m \right\|^2_2.
\end{align}
Setting $\boldsymbol{\hat{\theta}}_m = \boldsymbol{\theta}_m$ and $\boldsymbol{\tilde{\theta}}_m = \boldsymbol{\theta}_m + d\mathbf{e}_n$, and substituting (\ref{eq:gradient}) into (\ref{eq:quadratic upper bound}), we obtain
\begin{equation}
    \label{eq:term4}
	U_m(\boldsymbol{\theta}_m + d\mathbf{e}_n) \le U_m(\boldsymbol{\theta}_m) + \rho_{1}(\boldsymbol{\theta}_{m}) d + \frac{L_m}{2} d^2.
\end{equation}
Since $\rho_{2}(\boldsymbol{\theta}_m),$ $\rho_{3}(\boldsymbol{\theta}_m),$ and $\rho_{4}(\boldsymbol{\theta}_m)$ are continuous, and $\boldsymbol{\theta}_m$ is in the bounded region $[0,1]^N,$ the three functions are also bounded. That is, there exist a constant $\rho = 2(\bar{\rho}_{2}+\bar{\rho}_{3}+\bar{\rho}_{4})$ such that $\bar{\rho}_{2},$ $\bar{\rho}_{3},$ and $\bar{\rho}_{4}$ are three upper bounds:
\begin{equation}
	|\rho_{i}(\boldsymbol{\theta}_m)| \le \bar{\rho}_{i}, \quad i = 2,3,4.
\end{equation}
Since $d \in [-\theta_{m,n},\,1-\theta_{m,n}],$  we have
\begin{align}
\label{eq:term12}
d^2\frac{\rho}{2}+d^2 \rho_{2}(\boldsymbol{\theta}_{m}) + d^3 \rho_{3}(\boldsymbol{\theta}_{m}) + d^4 \rho_{4}(\boldsymbol{\theta}_{m}) \ge 0.
\end{align}
Therefore, when $\omega \ge L_m + \rho,$ substituting (\ref{eq:term12}) into (\ref{eq:distributed inexact}), we have
\begin{equation}
    \label{eq:term3}
U_m(\boldsymbol{\theta}_m)+\rho_{1}(\boldsymbol{\theta}_{m}) d + \frac{L_m}{2} d^2 \le U_m(\boldsymbol{\theta}_m) + p_{\mathrm{approx}}(d) + \frac{\omega}{2} d^2.
\end{equation}
By substituting (\ref{eq:term3}) into (\ref{eq:term4}), we obtain
\begin{equation}
\label{eq:term11}	U_m(\boldsymbol{\theta}_m + d\mathbf{e}_n) \le U_m(\boldsymbol{\theta}_m) + p_{\mathrm{approx}}(d) + \frac{\omega}{2} d^2,
\end{equation}
which means that $U_m(\boldsymbol{\theta}_m) + p_{\mathrm{approx}}(d) + \frac{\omega}{2} d^2$ is an upper bound of $U_m(\boldsymbol{\theta}_m + d\mathbf{e}_n)$. 

Next, let us check the property of the upper bound. For the right-hand side of (\ref{eq:term11}), we have
\begin{equation}
	\label{eq:term13}
	U_m(\boldsymbol{\theta}_m) + p_{\mathrm{approx}}(d) + \frac{\omega}{2} d^2 \le U_m(\boldsymbol{\theta}_m)+\rho_{1}(\boldsymbol{\theta}_{m}) d + \frac{\omega + \rho}{2} d^2,
\end{equation}
where the inequality derives from (\ref{eq:term12}). Therefore, letting $\tilde{d}$ denote the minimizer
of the right-hand side of (\ref{eq:term13}) within $[-\theta_{m,n}, 1-\theta_{m,n}]$, and noting that the solution of problem (\ref{eq:one-dim-approx}), i.e., $\bar{d}$, is actually the minimizer of the left-hand side of (\ref{eq:term13}) within $[-\theta_{m,n}, 1-\theta_{m,n}]$, we have
\begin{equation}
    \label{eq:term7}
U_m(\boldsymbol{\theta}_m)\!+\!p_{\mathrm{approx}}\big(\bar{d}\big) \!\!+ \!\!\frac{\omega}{2} \bar{d}^2 \!\!\le \!\!U_m(\boldsymbol{\theta}_m)\!+\!\rho_{1}(\boldsymbol{\theta}_{m}) \tilde{d} +\!\! \frac{\omega + \rho}{2} \tilde{d}^2,\\
\end{equation}
where $\tilde{d}$ can be obtained in a closed form:
\begin{equation}
    \label{eq:term5}
	\tilde{d} = \Pi_{[0,1]}\left(\theta_{m,n}- \frac{\rho_{1}(\boldsymbol{\theta}_{m})}{\omega + \rho}\right) - \theta_{m,n}.
\end{equation}
Due to the properties of the projection operator, we have
\begin{equation}
\label{eq:projection property}
	(x-\Pi_{[0,1]}(x))(y-\Pi_{[0,1]}(x))\le 0, ~\forall~y \in [0,1].
\end{equation}
Then, by substituting $x=\theta_{m,n}- \frac{\rho_{1}(\boldsymbol{\theta}_{m})}{\omega + \rho}$, and $y=\theta_{m,n}$ into (\ref{eq:term5}) and (\ref{eq:projection property}), we obtain
\begin{equation}
    \label{eq:term6}
	\rho_{1}(\boldsymbol{\theta}_{m})\tilde{d} \le - (\omega + \rho) \tilde{d}^2.
\end{equation}
Then, substituting (\ref{eq:term6}) into (\ref{eq:term7}) and (\ref{eq:term11}), we obtain
\begin{equation}
	\label{update_bound}
	U_m(\boldsymbol{\theta}_m + \bar{d}\mathbf{e}_n) \le U_m(\boldsymbol{\theta}_m) - \frac{\omega + \rho}{2} \tilde{d}^2. 
\end{equation}
This inequality demonstrates that the solution of problem (\ref{eq:one-dim-approx}) sufficiently decreases the objective function by $ (\omega + \rho)\tilde{d}^2 /2$.

Finally, we prove that Algorithm~\ref{alg:cd} guarantees the convergence to a stationary point of problem (\ref{distributed formulation}). We define a function $\chi(t)$ with respect to $t>0$ as follows:
\begin{equation}
    \label{eq:chi-function}
    \chi(t) = t^2 \left(\Pi_{[0,1]} \left(\theta_{m,n}- \frac{\rho_{1}(\boldsymbol{\theta}_{m})}{t}\right)- \theta_{m,n} \right)^2,
\end{equation}
which allows us to rewrite (\ref{update_bound}) as
\begin{equation}
    \label{eq:tighter-bound}
    U_m(\boldsymbol{\theta}_m + \bar{d}\mathbf{e}_n) \le U_m(\boldsymbol{\theta}_m) - \frac{\chi(\omega + \rho)}{2(\omega + \rho)}.
\end{equation}
Since $\chi(t)$ is an increasing function of $t$ \cite[Lemma 2.3.1]{bertsekas1997nonlinear}, and $ \omega + \rho \geq 1$, we have
\begin{equation}
    \label{eq:tighter-bound2}
    U_m(\boldsymbol{\theta}_m + \bar{d}\mathbf{e}_n) \le U_m(\boldsymbol{\theta}_m) - \frac{\chi(1)}{2(\omega + \rho)},
\end{equation}
where $ \chi(1) = \left(\Pi_{[0,1]} \left(\theta_{m,n}- \rho_{1}(\boldsymbol{\theta}_{m})\right)- \theta_{m,n} \right)^2.$ Next, we define 
\begin{equation}
	\label{eq:v-function-def}
	[\mathbf{v}(\boldsymbol{\theta}_{m})]=\Pi_{[0,1]}(\boldsymbol{\theta}_{m} - \nabla U_m(\boldsymbol{\theta}_m))  - \boldsymbol{\theta}_{m}.
\end{equation}
Then, using (\ref{eq:gradient}), we have
\begin{equation}
    \label{eq:v-function}
    [\mathbf{v}(\boldsymbol{\theta}_{m})]^2_n = \chi(1).
\end{equation}
Let $n_k$ denote the coordinate index selected at the $k$-th update, and let $\boldsymbol{\theta}^{k}_m$ represent the value after $k$ updates, i.e., $\boldsymbol{\theta}^{k+1}_m = \boldsymbol{\theta}^{k}_m + \bar{d}\,\bm{e}_{n_{k+1}}.$ Substituting (\ref{eq:v-function}) into (\ref{eq:tighter-bound2}), we have
\begin{equation}
\label{theta_difference}
U_m(\boldsymbol{\theta}^{k}_m) - U_m(\boldsymbol{\theta}^{k+1}_m) \ge \frac{[\mathbf{v}(\boldsymbol{\theta}^{k}_m)]^2_{n_{k+1}}}{2(\omega + \rho)}.
\end{equation}
 Summing (\ref{theta_difference}) over $k$ yields
\begin{equation}
	\label{theta_sum}
	U_m(\boldsymbol{\theta}_m^{0}) - U_m(\boldsymbol{\theta}_m^{\infty})  \ge \frac{1}{ 2(\omega + \rho)} \sum_{k=0}^{\infty} [\mathbf{v}(\boldsymbol{\theta}^{k}_m)]^2_{n_{k+1}}.
\end{equation}
Since $\boldsymbol{\theta}_m\!\in\![0,1]^N$ and 
$U_m(\boldsymbol{\theta}_m)$ is continuous on this compact set, there exists a
finite constant $U_m^{\mathrm{lb}}$ such that
$U_m(\boldsymbol{\theta}_m) \ge U_m^{\mathrm{lb}}\text{ for all }\,\boldsymbol{\theta}_m\in[0,1]^N$.
Hence, the left-hand side of (\ref{theta_sum}) is finite, which implies
\begin{equation}
	\label{d_summable}
	\lim_{k \to \infty} [\mathbf{v}(\boldsymbol{\theta}^{k}_m)]^2_{n_{k+1}} = 0.
\end{equation}
Since the coordinates  are randomly selected  at each iteration, the expected value of $[\mathbf{v}(\boldsymbol{\theta}^{k}_m)]^2_{n_{k+1}}$ is given by
\begin{equation}
	\label{eq:expected_value}
	\lim_{k \to \infty} \mathbb{E}[[\mathbf{v}(\boldsymbol{\theta}^{k}_m)]^2_{n_{k+1}}] =\lim_{k \to \infty} \frac{1}{N} \|\mathbf{v}(\boldsymbol{\theta}^{k}_m)\|_2^2=0.
\end{equation}
Let $\boldsymbol{\theta}_m^{\star}$ be a limit point of $\boldsymbol{\theta}^{k}_m$ generated by Algorithm~\ref{alg:cd}. Substituting (\ref{eq:expected_value}) into (\ref{eq:v-function-def}), we have
\begin{equation}
\label{eq:fixed_point_new}
\boldsymbol{\theta}_{m}^{\star}=\Pi_{[0,1]}(\boldsymbol{\theta}_{m}^{\star} - \nabla U_m(\boldsymbol{\theta}_m^{\star})).
\end{equation}
Then, let $\theta_{m,n}^{\star}$ be the $n$-th coordinate of $\boldsymbol{\theta}_m^{\star}$. By substituting $x=\theta_{m,n}^{\star} - [\nabla U_m(\boldsymbol{\theta}_m^{\star})]_n$ and $y=\tilde{\theta}_{m,n}$ into (\ref{eq:projection property}), and using (\ref{eq:fixed_point_new}), we obtain
\begin{equation}
[\nabla U_m(\boldsymbol{\theta}_m^{\star})]_n \,(\tilde{\theta}_{m,n}\!-\theta_{m,n}^{\star})\ge0,~\forall~\tilde{\theta}_{m,n}\!\!\in[0,1], n\!=\!1,2,...,N,
\end{equation}
which is the first-order optimality condition of problem~(\ref{sub b}). Hence, $\boldsymbol{\theta}_m^{\star}$ is a stationary point of problem~(\ref{sub b}).

\section{\textcolor{black}{Proof of Theorem~\ref{Distributed Convergency Theorem}}} \label{app3}
We first prove that $\mathcal {L}\left ({\left \{{\boldsymbol{\theta}_{m}^{(i)}}\right \}_{m=1}^{M}, \mathbf {a}^{(i)}; \left \{{\boldsymbol {\lambda }_{m}^{(i)}}\right \}_{m=1}^{M} }\right)$ is monotonically decreasing as $i$ increases. Since $\mu>2\tilde{L}_{m}$,  we have $\mu \mathbf {I}_{N}-\nabla ^{2} f_{m}(\boldsymbol{\theta}_{m})\succeq \tilde{L}_{m}\mathbf {I}_{N}$, which guarantees the strong convexity of problem (\ref{sub b}) with modulus $\tilde{L}_{m}$. Meanwhile, as demonstrated in Proposition~\ref{proposition1}, line 4 of Algorithm~\ref{distributed algorithm} can effectively solve problem (\ref{sub b}) to  a stationary point. Therefore, we have
\begin{align}
\label{D.2}
&\hspace {-3pc}\mathcal {L}\left({\left \{{\boldsymbol{\theta}_{m}^{(i-1)}}\right \}_{m=1}^{M}, \mathbf {a}^{(i-1)}; \left \{{\boldsymbol {\lambda }_{m}^{(i-1)}}\right \}_{m=1}^{M} }\right) \notag\\ 
&-\,\mathcal {L}\left ({\left \{{\boldsymbol{\theta}_{m}^{(i)}}\right \}_{m=1}^{M}, \mathbf {a}^{(i-1)}; \left \{{\boldsymbol {\lambda }_{m}^{(i-1)}}\right \}_{m=1}^{M} }\right) \notag\\
\geq&\sum _{m=1}^{M}\frac {\tilde{L}_{m}}{2}\left \Vert{ \boldsymbol{\theta}_{m}^{(i)}-\boldsymbol{\theta}_{m}^{(i-1)}}\right \Vert _{2}^{2}.
\end{align}

Moreover, since $\boldsymbol{\theta}_{m}^{(i)}$ is the stationary point of problem (\ref{sub b}), we have
\begin{align} 
\label{D.3}
\nabla f_{m}\left ({\boldsymbol{\theta}_{m}^{(i)}}\right) +\boldsymbol {\lambda }_{m}^{(i-1)} +{\mu }\left ({\boldsymbol{\theta}_{m}^{(i)}-\mathbf {a}^{(i-1)}}\right)=\mathbf {0},\notag\\\forall~m=1,\ldots,M. 
\end{align}
\noindent Substituting (\ref{dual ascent step}) into (\ref{D.3}), we have
\begin{align}
\label{D.4}
\left \Vert{ \boldsymbol {\lambda }_{m}^{(i-1)}-\boldsymbol {\lambda }_{m}^{(i)}}\right \Vert _{2}=&\left \Vert{ \nabla f_{m}\left ({\boldsymbol{\theta}_{m}^{(i-1)}}\right)-\nabla f_{m}\left ({\boldsymbol{\theta}_{m}^{(i)}}\right)}\right \Vert _{2} \notag\\
\leq&\tilde{L}_{m} \left \Vert{ \boldsymbol{\theta}_{m}^{(i-1)}-\boldsymbol{\theta}_{m}^{(i)}}\right \Vert _{2}.
\end{align}

\noindent Consequently, we have
\begin{align}
\label{D.5}
&\hspace {-3.5pc}\mathcal {L}\left ({\left \{{\boldsymbol{\theta}_{m}^{(i)}}\right \}_{m=1}^{M}, \mathbf {a}^{(i-1)}; \left \{{\boldsymbol {\lambda }_{m}^{(i-1)}}\right \}_{m=1}^{M} }\right) \notag\\ 
&-\,\mathcal {L}\left ({\left \{{\boldsymbol{\theta}_{m}^{(i)}}\right \}_{m=1}^{M}, \mathbf {a}^{(i-1)}; \left \{{\boldsymbol {\lambda }_{m}^{(i)}}\right \}_{m=1}^{M} }\right) \notag\\
=&\sum _{m=1}^{M} \left ({\boldsymbol {\lambda }_{m}^{(i-1)}-\boldsymbol {\lambda }_{m}^{(i)}}\right)^{\mathrm {T}} \left ({\boldsymbol{\theta}_{m}^{(i)}-\mathbf {a}^{(i-1)}}\right) \notag\\
=&-\frac {1}{\mu }\sum _{m=1}^{M}\left \Vert{ \boldsymbol {\lambda }_{m}^{(i-1)}-\boldsymbol {\lambda }_{m}^{(i)}}\right \Vert _{2}^{2} \notag\\
\geq&-\frac {1}{\mu }\sum _{m=1}^{M}\tilde{L}_{m}^{2}\left \Vert{ \boldsymbol{\theta}_{m}^{(i-1)}-\boldsymbol{\theta}_{m}^{(i)}}\right \Vert _{2}^{2},
\end{align}
where the second equality comes from (\ref{dual ascent step}) and the last inequality is due to (\ref{D.4}).

Letting $G_{1}\left ({\mathbf {a}}\right)\triangleq\sum \limits _{m=1}^{M} \left\{\boldsymbol {\lambda }_{m}^{\mathrm {T}}\left ({\boldsymbol{\theta}_{m}-\mathbf {a}}\right) +\frac {\mu }{2}\left \Vert{ \boldsymbol{\theta}_{m}-\mathbf {a}}\right \Vert_{2}^{2}\right\}$, we have
\begin{align}
\label{D.1}
&\hspace {-1pc}\mathcal {L}\left ({\left \{{\boldsymbol{\theta}_{m}^{(i)}}\right \}_{m=1}^{M}, \mathbf {a}^{(i-1)}; \left \{{\boldsymbol {\lambda }_{m}^{(i)}}\right \}_{m=1}^{M} }\right) \notag\\ 
&-\,\mathcal {L}\left ({\left \{{\boldsymbol{\theta}_{m}^{(i)}}\right \}_{m=1}^{M}, \mathbf {a}^{(i)}; \left \{{\boldsymbol {\lambda }_{m}^{(i)}}\right \}_{m=1}^{M} }\right) \notag\\
=&G_{1}\left ({\mathbf {a}^{(i-1)}}\right)-G_{1}\left ({\mathbf {a}^{(i)}}\right) \geq 0,
\end{align}
where the last inequality holds because $\mathbf {a}^{(i)}$ is the optimal solution of the convex quadratic problem (\ref{sub a}), and hence also the minimizer of $G_{1}\left ({\mathbf {a}}\right)$.

Combining (\ref{D.2}), (\ref{D.5}), and (\ref{D.1}), we obtain
\begin{align}
\label{D.6}
&\hspace {-1.5pc}\mathcal {L}\left ({\left \{{\boldsymbol{\theta}_{m}^{(i-1)}}\right \}_{m=1}^{M}, \mathbf {a}^{(i-1)}; \left \{{\boldsymbol {\lambda }_{m}^{(i-1)}}\right \}_{m=1}^{M} }\right) \notag\\ 
&-\,\mathcal {L}\left ({\left \{{\boldsymbol{\theta}_{m}^{(i)}}\right \}_{m=1}^{M}, \mathbf {a}^{(i)}; \left \{{\boldsymbol {\lambda }_{m}^{(i)}}\right \}_{m=1}^{M} }\right) \notag\\
\geq&\sum _{m=1}^{M}\left ({\frac {\tilde{L}_{m}}{2}-\frac {\tilde{L}_{m}^{2}}{\mu }}\right)\left \Vert{ \boldsymbol{\theta}_{m}^{(i)}-\boldsymbol{\theta}_{m}^{(i-1)}}\right \Vert _{2}^{2}\geq 0,
\end{align}
where the last inequality is due to $\mu >2\tilde{L}_{m}$. Therefore, $\mathcal {L}\left ({\left \{{\boldsymbol{\theta}_{m}^{(i)}}\right \}_{m=1}^{M}, \mathbf {a}^{(i)}; \left \{{\boldsymbol {\lambda }_{m}^{(i)}}\right \}_{m=1}^{M} }\right)$ is monotonically decreasing.

Next, we prove that $\mathcal {L}\left ({\left \{{\boldsymbol{\theta}_{m}^{(i)}}\right \}_{m=1}^{M}, \mathbf {a}^{(i)}; \left \{{\boldsymbol {\lambda }_{m}^{(i)}}\right \}_{m=1}^{M} }\right)$ is lower bounded. Substituting (\ref{dual ascent step}) into (\ref{D.3}), we have $\boldsymbol {\lambda }_{m}^{(i)}=-\nabla f_{m}\left ({\boldsymbol{\theta}_{m}^{(i)}}\right)$. Substituting this equation into (\ref{lagrangian}), we obtain
\begin{align}
\label{D.7}
&\hspace {-1.2pc}\mathcal {L}\left ({\left \{{\boldsymbol{\theta}_{m}^{(i)}}\right \}_{m=1}^{M}, \mathbf {a}^{(i)}; \left \{{\boldsymbol {\lambda }_{m}^{(i)}}\right \}_{m=1}^{M} }\right) \notag \\
=&\sum _{m=1}^{M}f_{m}\left ({\boldsymbol{\theta}_{m}^{(i)}}\right) -\sum _{m=1}^{M}\nabla f_{m}\left ({\boldsymbol{\theta}_{m}^{(i)}}\right)^{\mathrm {T}}\left ({\boldsymbol{\theta}_{m}^{(i)}-\mathbf {a}^{(i)}}\right)\notag \\
&+\frac {\mu }{2}\sum _{m=1}^{M}\left \Vert{ \boldsymbol{\theta}_{m}^{(i)}-\mathbf {a}^{(i)}}\right \Vert _{2}^{2}\notag\\
\overset{(a)}{\geq} &\sum _{m=1}^{M}f_{m}\left ({\mathbf {a}^{(i)}}\right)\geq\sum _{m=1}^{M}\log|\mathbf{C}_{m}|\overset{(b)}{\geq}\sum _{m=1}^{M}LK\log \varsigma _{m}^{2},
\end{align}
where the inequality (a) is derived from $\mu \mathbf {I}_{N}\succeq \tilde{L}_{m}\mathbf {I}_{N} \succeq \nabla ^{2} f_{m}\left ({\mathbf {x}_{m}}\right)$, and the inequality (b) is due to the fact that $\mathbf{C}_{m}$ contains the non-negative noise term $\varsigma^2_m\mathbf{I}_{LK}$ plus semipositve definited terms, so that $\mathbf{C}_{m} \succeq \varsigma^2_m\mathbf{I}_{LK}$. Therefore, $\mathcal {L}\left ({\left \{{\boldsymbol{\theta}_{m}^{(i)}}\right \}_{m=1}^{M}, \mathbf {a}^{(i)}; \left \{{\boldsymbol {\lambda }_{m}^{(i)}}\right \}_{m=1}^{M} }\right)$ is lower bounded.

Finally, we prove that any limit point of the sequence $\left ({\left \{{\boldsymbol{\theta}_{m}^{(i)}}\right \}_{m=1}^{M}, \mathbf {a}^{(i)}; \left \{{\boldsymbol {\lambda }_{m}^{(i)}}\right \}_{m=1}^{M} }\right)$ is a stationary point of problem (\ref{distributed formulation}). Summing (\ref{D.6}) over $i$ yields
\begin{align}
\label{D.10}
&\hspace {-1.5pc}\sum _{i=1}^{\infty} \sum _{m=1}^{M}\left ({\frac {\tilde{L}_{m}}{2}-\frac {\tilde{L}_{m}^{2}}{\mu }}\right)\left \Vert{ \boldsymbol{\theta}_{m}^{(i)}-\boldsymbol{\theta}_{m}^{(i-1)}}\right \Vert _{2}^{2} \notag\\
\leq&\mathcal {L}\left ({\left \{{\boldsymbol{\theta}_{m}^{(0)}}\right \}_{m=1}^{M}, \mathbf {a}^{(0)}; \left \{{\boldsymbol {\lambda }_{m}^{(0)}}\right \}_{m=1}^{M} }\right)\notag \\
-&\,\mathcal {L}\left ({\left \{{\boldsymbol{\theta}_{m}^{(\infty)}}\right \}_{m=1}^{M}, \mathbf {a}^{(\infty)}; \left \{{\boldsymbol {\lambda }_{m}^{(\infty)}}\right \}_{m=1}^{M} }\right) < \infty, 
\end{align}
where the last inequality comes from  (\ref{D.7}). Since $\mu >2\tilde{L}_{m}$, (\ref{D.10}) yields $\lim \limits _{i\rightarrow \infty }\left \Vert{ \boldsymbol{\theta}_{m}^{(i)}-\boldsymbol{\theta}_{m}^{(i-1)}}\right \Vert _{2}=0$. Together with (\ref{D.4}) and (\ref{dual ascent step}), we also have $\lim \limits _{i\rightarrow \infty }\left \Vert{ {\boldsymbol {\lambda }_{m}^{(i)}-\boldsymbol {\lambda }_{m}^{(i-1)}}}\right \Vert _{2}=\lim \limits _{i\rightarrow \infty }\left \Vert{ \boldsymbol{\theta}_{m}^{(i)}-\mathbf {a}^{(i)}}\right \Vert _{2}=0$. Thus, with $\left ({\left \{{\boldsymbol{\theta}_{m}^{*}}\right \}_{m=1}^{M}, \mathbf {a}^{*}; \left \{{\boldsymbol {\lambda }_{m}^{*}}\right \}_{m=1}^{M} }\right)$ denoting a limit point of the sequence $\left ({\left \{{\boldsymbol{\theta}_{m}^{(i)}}\right \}_{m=1}^{M}, \mathbf {a}^{(i)}; \left \{{\boldsymbol {\lambda }_{m}^{(i)}}\right \}_{m=1}^{M} }\right)$, we have
\begin{align}
\label{D.12}
\mathbf {a}^{(i-1)}\rightarrow&\mathbf {a}^{*},\quad \boldsymbol{\theta}_{m}^{(i-1)}\rightarrow \boldsymbol{\theta}_{m}^{*},~\boldsymbol {\lambda }_{m}^{(i-1)}\rightarrow \boldsymbol {\lambda }_{m}^{*}, \notag\\
\boldsymbol{\theta}_{m}^{*}=&\mathbf {a}^{*},\quad \forall~m=1,\ldots,M.
\end{align}
Taking limit for (\ref{D.3}), and applying (\ref{D.12}), we have
\begin{equation}
\nabla f_{m}\left ({\boldsymbol{\theta}_{m}^{*}}\right) +\boldsymbol {\lambda }_{m}^{*} =\mathbf {0},\quad \forall~m=1,\ldots,M.\label{D.13}
\end{equation}
On the other hand, since $\mathbf {a}^{(i)}$ is the optimal solution of problem (\ref{sub a}), it satisfies the following first-order optimality condition:
\begin{align}\label{D.14}
&\hspace {-.5pc}\mu \sum _{m=1}^{M}\left ({\mathbf {a}^{(i)}-\boldsymbol{\theta}_{m}^{(i)}}\right)^{\mathrm {T}}{\tilde {\mathbf {d}}}-\sum _{m=1}^{M}\left ({\boldsymbol {\lambda }_{m}^{(i)}}\right)^{\mathrm {T}}{\tilde {\mathbf {d}}}  \geq 0, \forall~{\tilde {\mathbf {d}}}\in \mathcal {T}\left ({\mathbf {a}^{(i)}}\right),
\end{align}
where $\mathcal {T}\left ({\mathbf {a}^{(i)}}\right)$ is the tangent cone of the feasible set of problem (15) at $\mathbf {a}^{(i)}$. Taking limit for (\ref{D.14}), and leveraging (\ref{D.12}), we have
\begin{align}
\label{D.15}
&\hspace {-1.75pc} -\sum _{m=1}^{M}\left ({\boldsymbol {\lambda }_{m}^{*}}\right)^{\mathrm {T}}{\tilde {\mathbf {d}}} \geq 0,\quad \forall~{\tilde {\mathbf {d}}}\in \mathcal {T}\left ({\mathbf {a}^{*}}\right),
\end{align}
where $\mathcal {T}\left ({\mathbf {a}^{*}}\right)\subseteq \mathcal {T}\left ({\mathbf {a}^{i}}\right)$.
Combining (\ref{D.12}), (\ref{D.13}), and (\ref{D.15}), we can conclude that $\left ({\left \{{\boldsymbol{\theta}_{m}^{*}}\right \}_{m=1}^{M}, \mathbf {a}^{*}; \left \{{\boldsymbol {\lambda }_{m}^{*}}\right \}_{m=1}^{M} }\right)$ is a stationary point of problem (\ref{distributed formulation}).

\bibliographystyle{IEEEtran}
\bibliography{IEEEabrv,references}

\end{document}